%% file: networkslicing_decomposition.tex
\documentclass[10pt,twocolumn,twoside]{IEEEtran}

\usepackage{amsmath,amssymb,amsthm}
\usepackage{algorithmic}
\usepackage{pgf,tikz,pgfplots}
\usetikzlibrary{arrows,decorations.markings,automata}
\usepackage{bm}
\usepackage{color}
\usepackage{float}
\usepackage{makecell}
\usetikzlibrary{arrows,decorations.markings,automata,3d,fit, shapes.geometric, positioning, arrows.meta}

\usepackage{mathtools}

\usepackage{subfigure}
\usepackage{cases}
\usepackage{tabularx}
\usepackage{amsthm}
\usepackage{multicol}
\usepackage{todonotes}
\usepackage{xspace}
\usepackage{hyperref}
\hypersetup{
	colorlinks=true,
	linkcolor=blue,
	filecolor=magenta,
	urlcolor=cyan,
	citecolor=blue
}

\usepackage{cleveref}
\newtheorem{theorem}{Theorem}
\newtheorem{lemma}[theorem]{Lemma}
\newtheorem{proposition}[theorem]{Proposition}

\newtheorem{example}{Example}

\crefname{theorem}{Theorem}{Theorems}
\crefname{example}{Example}{Examples}
\crefname{observation}{Observation}{Observations}
\crefname{remark}{Remark}{Remarks}
\crefname{proposition}{Proposition}{Propositions}
\crefname{lemma}{Lemma}{Lemmas}
\crefname{corollary}{Corollary}{Corollaries}
\crefname{fact}{Fact}{Facts}
\crefname{algocf}{Algorithm}{Algorithms}	
\crefname{table}{Table}{Tables}	
\crefname{figure}{Fig.}{Figs.}
\crefname{algorithm}{Algorithm}{Algorithms}
\crefname{section}{Section}{Sections}

\newcommand{\bb}{\mathbf{b}}

\newcommand{\rev}[1]{{\color{black}{#1}}}
\newcommand{\revv}[1]{{\color{black}{#1}}}
\newcommand{\revvv}[1]{{\color{black}{#1}}}

\newcommand{\red}[1]{{\color{red}{#1}}}
\newcommand{\PSUMR}{{PSUM-R}\xspace}

\newcommand{\LPoR}{{LPoR}\xspace}
\newcommand{\LPdR}{{LPdR}\xspace}

\newcommand{\SOTA}{{SOTA}\xspace}
\newcommand{\MBLP}{{MBLP}\xspace}
\newcommand{\EXACT}{{MBLP-S}\xspace}

\newcommand{\NS}{{NS}\xspace}
\newcommand{\FP}{{FP}\xspace}
\newcommand{\TR}{{TR}\xspace}
\newcommand{\CG}{{CG}\xspace}

\newcommand{\CBD}{{CBD}\xspace}
\newcommand{\CBDs}{{CBDs}\xspace}
\newcommand{\VNE}{{VNE}\xspace}
\newcommand{\IterMax}{\text{IterMax}\xspace}
\newcommand{\LP}{{\text{LP}}\xspace}
\newcommand{\LIN}{{\text{L}}\xspace}

\DeclareMathOperator*{\minimize}{minimize}

\newcommand{\network}{fish-20-9}
\usepackage{todonotes}

\newcommand{\CL}{\mathcal{L}}
\newcommand{\V}{\mathcal{V}}
\newcommand{\I}{\mathcal{I}}
\newcommand{\F}{\mathcal{F}}
\newcommand{\K}{\mathcal{K}}
\newcommand{\X}{\mathcal{X}}

\def\x{\boldsymbol{x}}
\def\y{\boldsymbol{y}}
\def\z{\boldsymbol{z}}
\def\w{\boldsymbol{w}}
\def\r{\boldsymbol{r}}

\def\d{\boldsymbol{d}}
\def\bb{\boldsymbol{b}}
\def\a{\boldsymbol{\alpha}}
\def\b{\boldsymbol{\beta}}

\definecolor{ududff}{rgb}{0.30196078431372547,0.30196078431372547,1.}
\usepackage{algorithm,algorithmic}

\makeatletter
\newcommand\Label[1]{&\refstepcounter{equation}(\theequation)\ltx@label{#1}&}
\makeatother

\newcommand{\figuresize}{1.7in}

\begin{document}
	\setlength{\jot}{0.09cm}

\title{
	\rev{An Efficient Benders Decomposition Approach for Optimal Large-Scale Network Slicing} \thanks{Part of this work \cite{Chen2023a} has been presented at the 24th IEEE International Workshop on Signal Processing Advances in Wireless Communications (SPAWC), Shanghai, China, September 25-28, 2023.}
}
\author{\IEEEauthorblockN{Wei-Kun Chen, Zheyu Wu, Rui-Jin Zhang, Ya-Feng Liu,  Yu-Hong Dai, and Zhi-Quan Luo}
	\thanks{
	\revvv{The work of W.-K. Chen was supported in part by National Natural Science Foundation of China (NSFC) under Grant 12101048. 
	The work of Z. Wu and Y.-F. Liu was supported in part by the NSFC under Grant 12371314 and Grant 11991020.
	The work of Y.-H. Dai was supported in part by the NSFC under Grant 11991021 and Grant 12021001.
	The work of Z.-Q. Luo was supported by the Guangdong Major Project of Basic and Applied Basic Research (No. 2023B0303000001), the Guangdong Provincial Key Laboratory of Big Data Computing, and the National Key Research and Development Project under grant 2022YFA1003900.
	(\emph{Corresponding author: Ya-Feng Liu.})}
	}
	\thanks{
		W.-K. Chen is with the School of Mathematics and Statistics/Beijing Key Laboratory on MCAACI, Beijing Institute of Technology, Beijing 100081, China (e-mail: chenweikun@bit.edu.cn).
		%
		Z. Wu, R.-J. Zhang, Y.-F. Liu, and Y.-H. Dai are with the State Key Laboratory
		of Scientific and Engineering Computing, Institute of Computational Mathematics and Scientific/Engineering Computing, Academy of Mathematics and Systems Science, Chinese Academy of Sciences, Beijing 100190, China (e-mail: \{wuzy, zhangruijin, yafliu, dyh\}@lsec.cc.ac.cn).
		%
		Z.-Q. Luo is with the Shenzhen Research Institute of Big Data and The Chinese University of Hong Kong, Shenzhen 518172, China (e-mail: luozq@cuhk.edu.cn)
	}
}

\maketitle

\begin{abstract}
	This paper considers the network slicing (\NS) problem which attempts to map multiple customized virtual network requests to a common shared network infrastructure and allocate network resources to meet diverse service requirements.
	This paper proposes an efficient \rev{customized Benders} decomposition algorithm for globally  solving  the large-scale NP-hard \NS problem.
	The proposed algorithm decomposes the hard \NS problem into two relatively easy function placement (\FP) and traffic routing (\TR) subproblems and iteratively solves them enabling the information feedback between each other, which makes it particularly suitable to solve large-scale problems.
	Specifically, the \FP subproblem is to place service functions into cloud nodes  in the network, and solving it can return a function placement strategy based on which the \TR subproblem is defined;
	and the \TR subproblem is to find paths connecting two nodes hosting two adjacent functions in the network, and solving it can either verify that the solution of the \FP subproblem {is} an optimal solution of the original problem, or return a valid inequality to the \FP subproblem that cuts off the current infeasible solution.
	The proposed algorithm is guaranteed to find the \rev{globally optimal} solution of the \NS problem.
	By taking the special structure of the \NS problem into consideration, we successfully develop two families of valid inequalities that render the proposed algorithm converge much more quickly and thus much more efficient.
	\revv{Numerical results demonstrate that the proposed valid inequalities effectively accelerate the convergence of the decomposition algorithm, and the proposed algorithm significantly outperforms the existing algorithms in terms of both solution efficiency and quality.}
\end{abstract}

\input{section_introduction}

\input{section_problem}

\input{section_decomposition}

\input{section_ineq}

\input{section_numres2}
\input{appendix}

\input{references}

\end{document}

%% file: section_introduction.tex
\section{Introduction}
\label{sec:intro}

The fifth generation (5G) and the future sixth generation (6G) networks are expected to simultaneously support multiple services with diverse requirements such as peak data rate, latency, reliability, and energy efficiency \cite{Vassilaras2017}.
In traditional networks, service requests (consisting of a prespecified sequence of service functions) are implemented by dedicated hardware in fixed locations, which is inflexible and inefficient \cite{Mirjalily2018}.
In order to improve the resource provision flexibility and efficiency of the network, a key enabling technology called network function virtualization (NFV) has been proposed \cite{Mijumbi2016}.
In contrast to traditional networks, NFV-enabled networks can efficiently leverage virtualization technologies to configure some specific cloud nodes in the network to process service functions on-demand, and establish customized virtual networks for all services.
However, since the functions of all services are implemented over a single shared network infrastructure, it is crucial to efficiently allocate network (e.g., cloud and communication) resources to meet the diverse service requirements, subject to the capacity constraints of all cloud nodes and links in the network.
In the literature, the above resource allocation problem is called \emph{network slicing} (NS).


\subsection{Literature Review} 
Various algorithms have been proposed in the literature to solve the \NS problem or its variants; see \cite{Jarray2012}-\cite{Promwongsa2020} and the references therein.
To streamline the literature review, we classify the existing approaches into the following two categories: (i) exact algorithms that find an optimal solution for the problem and (ii) heuristic algorithms that aim to find a (high-quality) feasible solution for the problem.

In particular, the works \cite{Jarray2012}-\cite{Domenico2020} proposed the link-based mixed binary linear programming (\MBLP) formulations for the \NS problem and employed standard \MBLP solvers like CPLEX to solve their problem formulations to find an optimal solution.
However, when applied to solve large-scale \NS problems, the above approaches generally suffer from low computational efficiency due to the large search space and problem size. 
The works \cite{Jarray2015}-\cite{Yang2021} proposed the path-based \MBLP formulations for the \NS problems. 
However, due to the exponential number of path variables, the state-of-the-art (\SOTA) \MBLP solvers cannot solve these formulations efficiently.

To quickly obtain a feasible solution of the \NS problem, various heuristic and meta-heuristic algorithms have also been proposed in the literature.
In particular, 
\revv{the works} \cite{Chowdhury2012}-\cite{Luizelli2015} proposed two-stage heuristic algorithms that first decompose the link-based \MBLP formulation of the \NS problem into a function placement (\FP) subproblem and a traffic routing (\TR) subproblem and then solve each subproblem in a one-shot fashion.
The \FP and \TR subproblems attempt to map service functions into cloud nodes
and find paths connecting two nodes hosting two adjacent functions in the network, respectively.
To solve the \FP subproblem, \rev{the work \cite{Chowdhury2012}} first solved the linear programming (LP) relaxation of the \NS problem and \rev{then used a one-shot rounding strategy; the work \cite{Chen2023} used an  LP dynamic rounding  algorithm that sequentially rounds the value of a variable {so that its value} is consistent to other already rounded variables; the work \cite{Zhang2017} first solved a series of \LP relaxation problems to find a near integral solution and then used a one-shot rounding  strategy; and the works \cite{Yu2008}-\cite{Luizelli2015} used some greedy heuristics (without solving any LP).}
%
Once a solution for the mapping of service functions and cloud nodes is obtained, the \TR subproblem can be solved by using shortest path, $k$-shortest path, or multicommodity flow algorithms.
Unfortunately, solving the \FP subproblem in a heuristic manner without taking the \TR subproblem into account can lead to infeasibility or low-quality solutions.
The work \cite{Mohammadkhan2015} divided the services into several groups to obtain several ``small'' subproblems, and sequentially solved the subproblems using standard \MBLP solvers. 
Similarly, as this approach failed to simultaneously take the function placement and traffic routing of all services into account, it is likely to lead to infeasibility or return low-quality solutions. 
The works \cite{Jarray2015}-\cite{Gupta2018} proposed column
generation (\CG) approaches \cite{Conforti2014} to solve the path-based \MBLP formulation of the \NS problem.
Unfortunately, only a subset of the paths is considered for the traffic routing of flows in the \CG approach, which generally leads to infeasibility or low-quality solutions.
Meta-heuristic algorithms were also developed to obtain a feasible solution for the \NS problem including particle swarm optimization \cite{Zhang2012}, simulated annealing \cite{Li2015}, and  Tabu search \cite{Abu-Lebdeh2017,Promwongsa2020}.

In summary, the existing approaches to solving the \NS problem either suffer from low efficiency (as they need to call \MBLP solvers to solve \MBLP formulations with a large problem size) or return a low-quality solution (as they fail to take the information of the whole problem into consideration). 
\rev{The main motivation of this paper is to develop an \emph{efficient} algorithm that takes the information of the whole problem into consideration and finds a \emph{globally optimal} solution for the \emph{large-scale} \NS problem.}
\revv{Integer programming and Benders decomposition techniques \cite{Conforti2014} play vital roles in the developed algorithm.}

\subsection{Our contributions}

\rev{The main contribution of this paper is the proposed customized Benders decomposition \revvv{(\CBD)} algorithm for solving large-scale NS problems,}
{which decomposes the hard large-scale \NS problem into two relatively easy \FP and \TR subproblems and iteratively
solves them enabling the information feedback between each other.
Two key features of the proposed algorithm are: 
(i) It is guaranteed to find an optimal solution of the NS problem if the number of iterations between the \FP and \TR subproblems is allowed to be sufficiently large. 
Even though the number of iterations between the subproblems is small (e.g., 5), the proposed algorithm can still return a much better solution than existing two-stage heuristic algorithms \cite{Chowdhury2012}-\cite{Luizelli2015}.
(ii) The two subproblems in the proposed decomposition algorithm are much easier to solve than the original problem, 
which renders it particularly suitable to solve large-scale \NS problems, compared with existing algorithms based on LP relaxations \cite{Chowdhury2012,Chen2023,Zhang2017} and algorithms needed to call \MBLP solvers \cite{Chen2021}.    

Two technical contributions of the paper are summarized as follows.}
\begin{itemize}
	\item Different from existing two-stage heuristic algorithms \cite{Chowdhury2012}-\cite{Luizelli2015} that solve the \FP  and \TR subproblems in a one-shot fashion, our proposed decomposition algorithm solves the two subproblems in an iterative fashion enabling the information feedback between them.
	The information feedback between the two subproblems plays a central role in the \rev{proposed} algorithm, which is the first technical contribution of the paper. 
	In particular, solving the \FP subproblem can return a function placement strategy based on which the TR subproblem is defined; 
and solving the \TR subproblem can either verify that the solution of the \FP subproblem {is} an optimal solution of the original problem or return a valid inequality to the \FP subproblem that cuts off the current infeasible solution by virtue of Farkas' lemma.
\item We propose two families of valid inequalities that judiciously exploit and utilize the special structures of the \NS problem into consideration
such as the connectivity and limited link capacity structures of the network.
The proposed inequalities, as the second technical contribution of this paper,  can be directly added in the \FP subproblem, thereby significantly reducing the gap between the \FP problem and the \NS problem.
Consequently, the \rev{proposed algorithm} with these valid inequalities converges much faster and thus is much more efficient than that without the corresponding inequalities.
\end{itemize}

Numerical results demonstrate that the proposed valid inequalities significantly reduce the gap between the \FP and \NS problems and  accelerate  the convergence of the decomposition algorithm; 
and the proposed decomposition algorithm significantly outperforms the existing \SOTA algorithms  \cite{Chen2021,Chowdhury2012,Chen2023,Zhang2017} in terms of both solution efficiency and quality.

In our prior work \cite{Chen2023a}, we presented a \rev{\CBD} framework for solving the \NS problem. 
This paper is a significant extension of \cite{Chen2023a} towards a much more efficient \rev{\CBD} algorithm.
In particular, we develop two families of valid inequalities to accelerate the convergence of the \rev{\CBD} algorithm.
Moreover, we provide detailed computational results showing the effectiveness of the developed valid inequalities in reducing the gap and improving the performance of the \rev{\CBD} algorithm. 
The development of valid inequalities and related simulation results are completely new compared to our prior work.

The rest of the paper is organized as follows. 
\cref{sec:modelformulation} introduces the system model and
mathematical formulation for the \NS problem. 
\cref{sect:TSD,sect:ineq} present the \rev{\CBD} algorithm and the valid inequalities to speed up the algorithm, respectively.
\cref{sect:numres} reports the numerical results. 
Finally, \cref{sect:conclusion} draws the conclusion.

%% file: section_problem.tex
\section{System Model and Problem Formulation}
\label{sec:modelformulation}

Let $\mathcal{G}=\{\mathcal{I},\mathcal{L}\}$ denote the substrate (directed) network, where $\mathcal{I}=\{i\}$ and $\mathcal{L}=\{(i,j)\}$ are the sets of nodes and links, respectively. 
Let $ \mathcal{V} \subseteq \mathcal{I}$ be the set of cloud nodes that can process functions.
As assumed in \cite{Zhang2017}, processing one unit of data rate consumes one unit of (normalized) computational capacity, and the total data rate on cloud node $ v \in \V$ must be upper bounded by its capacity $\mu_v$.
Similarly, each link  $ (i,j) \in \CL $ has a communication capacity $ C_{ij} $.
There is a set of services $\mathcal{K}=\{k\}$ that need to be supported by the network.
Each service $ k\in \mathcal{K} $ relates to a customized service function chain (SFC) consisting of $ \ell_k $ service functions that have to be processed in sequence by the network: $f_{1}^k\rightarrow f_{2}^k\rightarrow \cdots \rightarrow f_{\ell_k}^k$ \cite{Zhang2013}.
\rev{If a function is processed by multiple nodes, it may introduce the overhead of maintaining the network state consistency between different parts of the function \cite{Paschos2018,Nguyen2022} (as they need to communicate with each other to exchange data or synchronize their states). 
In order to reduce the  coordination overhead, we follow \cite{Zhang2017,Woldeyohannes2018,Paschos2018,Nguyen2022} to assume that each function must be processed at exactly one cloud node.}
For service $ k $, let $S^k, D^k\notin \mathcal{V}$ denote the source and destination nodes, and let $ \lambda^k_0 $ and $ \lambda^k_s $, $s\in \F^k :=\{1, 2, \ldots, \ell_k\}$, denote the data rates before receiving any function and after receiving {function} $ f^k_s $, respectively.
\rev{In this paper, we assume that the service functions $\{f_s^k\}$ and their data rates $\{\lambda_s^k\}$ are known before mapping them into the substrate network, which is a common assumption in the literature \cite{Vassilaras2017,Liu2017,Woldeyohannes2018}.}

The \NS problem is to determine the function placement, the routes, and the associated data rates on the corresponding routes of all services while satisfying the capacity constraints of all cloud nodes and links.
\revv{Next, we shall introduce the problem formulation.}
\\[2pt]
{\bf\noindent$\bullet$ Function Placement\\[2pt]}
\indent Let $x^{k,s}_v=1$ indicate that function $f^k_s$ is processed by cloud node $v$; otherwise, $x^{k,s}_v=0$.
The following constraint \eqref{onlyonenode} ensures that each function $f_s^k$ must be processed by exactly one cloud node:
\begin{eqnarray}
\label{onlyonenode}
\sum_{v\in \mathcal{V}}x^{k,s}_v=1,~\forall~s\in  \mathcal{F}^k,~\forall ~k \in \mathcal{K}.
\end{eqnarray}
Let $y_v\in \{0,1\}$ denote whether or not cloud node $v$ is activated and powered on.
By definition, we have 
\begin{equation}
\label{xyxelation}
x^{k,s}_v \leq y_v, ~ \forall~v \in \mathcal{V},~\forall~s \in \mathcal{F}^k,~\forall~k \in \mathcal{K}.
\end{equation}
The total data rates on cloud node $v$ is upper bounded by $\mu_v$:
\begin{equation}
\label{nodecapcons}
\sum_{k\in \mathcal{K}}\sum_{s \in \F^k}\lambda_s^k x^{k,s}_v\leq \mu_v y_v,~\forall~ v \in \mathcal{V}.
\end{equation}
{\bf\noindent$\bullet$ Traffic Routing\vspace{0.1cm}\\}
\indent We use $ (k,s) $, $1 \leq s < \ell_k$, to denote the traffic flow which is routed between the two cloud nodes hosting the two adjacent functions $ f_s^k $ and $ f_{s+1}^k $ (with the \rev{data} rate being $\lambda_{s}^k$). 
Similarly, $(k,0)$ denotes the traffic flow which is routed between the source $S^k$ and the cloud node hosting function $f_1^k$ (with the \rev{data} rate being $\lambda_{0}^k$)\rev{;}
$(k,\ell_k)$ denotes the traffic flow which is routed between the cloud node hosting function $f_{\ell_k}^k$ and the destination $D^k$ (with the date rate being $\lambda_{\ell_k}^k$).
Let $ r_{ij}^{k,s} $ be the fraction of \rev{data} rate $\lambda_{s}^k$ of flow $(k,s)$ on link $(i,j)$.
Then the link capacity constraints can be written as follows:
\begin{equation}
	\label{linkcapcons}
	\sum_{k \in \mathcal{K}} \sum_{s\in \mathcal{F}^k \cup \{0\}}\lambda_{s}^k r_{ij}^{k,s} \leq C_{ij}, ~  \forall~(i,j) \in \mathcal{L}.
\end{equation}

To ensure that all functions of each flow $k$ are processed in the predetermined order $f_{1}^k\rightarrow f_{2}^k\rightarrow \cdots \rightarrow f_{\ell_k}^k$, we need the flow conservation constraint \cite{Chen2023}:
\begin{align}
	& \sum_{j: (j,i) \in \mathcal{{L}}} r_{ji}^{k,s} - \sum_{j: (i,j) \in \mathcal{{L}}} r_{ij}^{k,s} =b_i^{k,s}(\x),\nonumber\\[-2pt]
	& \qquad \qquad\qquad\qquad \forall~ i \in \mathcal{I},~\forall~s \in \mathcal{F}^k\cup \{0\}, ~\forall~k \in \mathcal{K},\label{SFC}
\end{align}
where 
\begin{equation*}
	b_i^{k,s}(\x)=\left\{\begin{array}{ll}
		-1,&\text{if~}s=0~\text{and}~i= S^k;\\[3pt]
		x^{k,s+1}_{i},&\text{if~}s=0~\text{and}~i\in \mathcal{V};\\[3pt]
		x_i^{k,s+1}-x_i^{k,s},&\text{if~}1\leq s< \ell_k
		~\text{and}~i\in \mathcal{V};\\[3pt]
		-x_i^{k,s},&\text{if~}s=\ell_k
		~\text{and}~i\in \mathcal{V};\\[3pt]
		1,&\text{if~}s=\ell_k
		~\text{and}~i=D^k;\\[3pt]
		0,& \text{otherwise}.
	\end{array}
	\right.
\end{equation*}
Note that the right-hand side of the above flow conservation constraint \eqref{SFC} depends on the type of node $i$  (e.g., a source node, a destination node, an activated/inactivated cloud node, or an intermediate node) and the traffic flow $(k,s)$.
Let us look into the second case where $b_{i}^{k,s}(\x)=x_{i}^{k,1}$: if $x_i^{k,1}=0$, then \eqref{SFC} reduces to the classical flow conservation constraint; if $x_i^{k,1}=1,$ then \eqref{SFC} reduces to 
$$ \sum_{j: (j,i) \in \mathcal{{L}}} r_{ji}^{k,0} - \sum_{j: (i,j) \in \mathcal{{L}}} r_{ij}^{k,0} = 1,$$
which enforces that the difference of the data rates of flow $(k,0)$ coming into node $i$ and going out of node $i$ should be equal to $\lambda_0^k$.\\[2pt]  
{\bf\noindent$\bullet$ Problem Formulation\\[2pt]}
\indent The power consumption of a cloud node is the combination of the static power
consumption and the dynamic load-dependent power consumption
(that increases linearly with the load) \cite{3gpp}.
The \NS problem is to minimize the total power
consumption of the whole network:
\begin{equation}
	\label{ns}
	\tag{NS}
\begin{aligned} \revv{\minimize_{\substack{\boldsymbol{x},\,\boldsymbol{y},\,\boldsymbol{r}}}}~&  \sum_{v \in \mathcal{V}}p_v y_v + \sum_{v\in\mathcal{V}} \sum_{k \in \mathcal{K}} \sum_{s \in \F^k} c^{k,s}_v x^{k,s}_v \\
{\revv{\text{subject to}}~}&  \eqref{onlyonenode}\text{--}\eqref{rbounds},  \\
\end{aligned}
\end{equation}
where 
\begin{align}
	& y_v\in\{0,1\},~\forall~v\in\mathcal{{V}},\label{ybounds}\\
	&  x^{k,s}_v \in \{0,1\},~\forall ~v\in\mathcal{{V}},~\forall~s\in \F^k,~\forall~k\in\mathcal{K},\label{xbounds} \\
	& r_{ij}^{k,s} \geq 0,~\forall~(i,j)\in \mathcal{L},~\forall~s\in \F^k\cup \{0\},~\forall~k\in \mathcal{K}.\label{rbounds}
\end{align}
Here $p_v$ is the power consumption of cloud node $v$ (if it is activated) and $c_v^{k,s}$ is the power consumption of placing function $f_s^k$ into cloud node $v$. 

Problem \eqref{ns} is an \MBLP problem \cite{Zhang2017} and therefore can be solved to global optimality using the \SOTA~\MBLP solvers, e.g.,  Gurobi, CPLEX, and SCIP.
However, due to the intrinsic (strong) NP-hardness  \cite{Zhang2017}, and particularly the large problem size,
the above approach cannot solve large-scale \NS problems efficiently.
In the next section, we shall develop an efficient decomposition algorithm for solving large-scale \NS problems.

%% file: section_decomposition.tex
\section{An Efficient  \rev{\CBD} Algorithm}
\label{sect:TSD}

\rev{Benders decomposition is an algorithmic framework that can globally solve large-scale mixed integer linear programming problems \cite{Conforti2014}. Its basic idea is to decompose the original problem into two simpler problems, called \emph{Benders master} and \emph{subproblems}, and iteratively solve the two subproblems until an optimal solution of the original problem is found.
{In this section}, we shall develop an efficient \CBD algorithm for solving the large-scale \NS problem based on its special decomposition structure.
We shall also provide some analysis results on the proposed algorithm.}

\subsection{Proposed Algorithm}
%

Before going into details, we \rev{first} give a high-level preview of the proposed decomposition algorithm
\rev{for solving problem \eqref{ns}. From \cref{sec:modelformulation}, \revv{we can see the decomposition structure of problem  \eqref{ns}.} More specifically,}
problem \eqref{ns} is a combination of the \FP subproblem and the \TR subproblem and the two subproblems are deeply interwoven with each other via constraints \rev{in} \eqref{SFC}, where the placement variable $\x$ and the traffic routing variable $\boldsymbol{r}$ are coupled.
The \FP subproblem addresses the placement of functions into cloud nodes in the network while the \TR subproblem addresses the traffic routing of all pairs of two cloud nodes hosting two adjacent functions.
The basic idea of the proposed decomposition algorithm is \rev{to
first \revv{exploit the decomposition structure of problem \eqref{ns}} and}
decompose the ``hard'' problem \eqref{ns} into two relatively ``easy'' subproblems, \rev{and then to} iteratively solve them with a useful information feedback between the subproblems (to deal with the coupled constraints) until an optimal solution of problem \eqref{ns} is found.
\rev{The \FP and \TR subproblems {correspond} to the Benders master and subproblems, respectively.} 
In particular, solving the \FP subproblem can provide a solution $\bar{\x}$ to the \TR subproblem based on which the \TR subproblem is defined; and 
solving the \TR subproblem can either verify that  $\bar{\x}$ is an optimal solution of problem \eqref{ns},
or return a valid inequality of the form $\boldsymbol{\pi}^\top \x \geq \pi_0$ \rev{(called \emph{Benders cut})} for the \FP subproblem that cuts off the current infeasible solution $\bar{\x}$. 
\rev{Here, $\boldsymbol{\pi}$ and $\pi_0$ are the coefficient vector of the variable $\x$ and the right-hand side of the valid inequality, respectively, which can be obtained by solving the \TR subproblem.}
The information feedback between the \FP and \TR subproblems in the \rev{proposed} algorithm is illustrated in \cref{ideaTSD}. 
In the following, we shall present \rev{the algorithm} in detail.

\begin{figure}\centering
	\begin{tikzpicture} [->,>=stealth',shorten >=0pt,auto,node distance=3.5cm,
		semithick]
	\tikzstyle{arrow}=[->,>=stealth, thick]
	\tikzstyle{juxing1} = [rounded corners,inner sep = 2pt,draw,rectangle,line width = 1pt,fill = white,minimum width = 2.4cm, minimum height = 1cm]
	\tikzstyle{juxing2} = [rounded corners,minimum height=1.9cm,draw,rectangle,line width = 1pt,fill = white,minimum width = 13.5cm, style = dotted]
	\tikzstyle{juxing3} = [rounded corners,minimum height = 1.9cm,draw,rectangle,line width = 1pt,fill = white,minimum width = 3.9cm, style = dotted]
		\node (1) [juxing1,align=center,fill=blue!20] at (0,0) {FP subproblem};
		\node (3) [juxing1,right  = 4cm of 1,align=center,fill=blue!20] {\TR subproblem}; 
		\path
		(1) edge[bend left]              node {\makecell{A solution $\bar{\x}$ based on which\\ the \TR subproblem is defined}} (3)
		(3) edge[bend left,below]              node {\makecell{A certificate that $\bar{\x}$ {is} an optimal solution of \eqref{ns},\\ or a valid inequality $\boldsymbol{\pi}^\top \x \geq \pi_0$ that cuts off $\bar{\x}$}} (1);
\end{tikzpicture}
\caption{The information feedback between the \FP and \TR subproblems in the decomposition algorithm.}
\label{ideaTSD}
\end{figure}
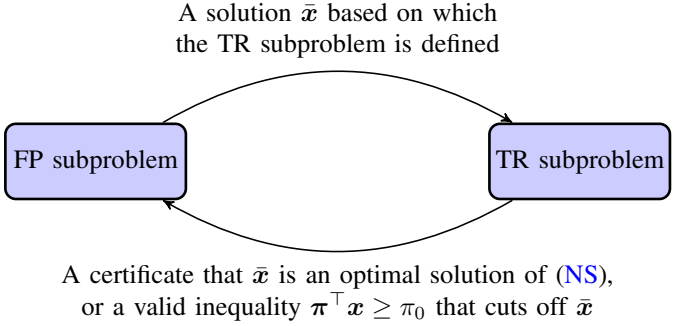

\subsubsection{\FP subproblem}
The \FP subproblem addresses the placement of functions into cloud nodes in the network, and can be presented as 
\begin{equation}
	\label{vnf}
	\tag{FP}
	\begin{aligned} \revv{\minimize_{\substack{\x,\,\y}}}~&  \sum_{v \in \mathcal{V}}p_v y_v + \sum_{v\in\mathcal{V}} \sum_{k \in \mathcal{K}} \sum_{s \in \F^k} c^{k,s}_v x^{k,s}_v  \\
	{\revv{\text{subject to}}~}&  \eqref{onlyonenode}\text{--}\eqref{nodecapcons},~\eqref{ybounds}\text{--}\eqref{xbounds}.
	\end{aligned}
\end{equation}
Problem \eqref{vnf} is an \MBLP problem and is a relaxation of  problem \eqref{ns} 
in which all traffic routing related variables (i.e., $\r$) and constraints (i.e., \eqref{linkcapcons}, \eqref{SFC}, and \eqref{rbounds}) are dropped. 
Note that the dimension of problem \eqref{vnf} is significantly smaller than that of problem \eqref{ns}\rev{; see next subsection for a detailed discussion.} 
Therefore, solving problem \eqref{vnf} by the \SOTA~\MBLP solvers is significantly more efficient than solving problem \eqref{ns}.
To determine whether a 
solution $(\bar{\x}, \bar{\y})$ of problem \eqref{vnf} can define an optimal solution of problem \eqref{ns}, i.e., whether there exists a point $\bar{\r}$ such that $(\bar{\x}, \bar{\y}, \bar{\r})$ is a feasible solution of problem \eqref{ns}, we need to solve the \TR subproblem.

\subsubsection{\TR subproblem}
Given a solution $(\bar{\x}, \bar{\y})$, the \TR subproblem aims to find a traffic routing strategy to route the traffic flow $(k,s)$ for all $k \in \K$ and $s \in \F(k)\cup \{0\}$ in the network, or prove that no such strategy exists.
This is equivalent to determining whether  
the linear system defined by \eqref{linkcapcons}, \eqref{SFC} with $\x= \bar{\x}$, and \eqref{rbounds}, i.e., 
\begin{equation}
	\tag{TR}
	\label{tr}
	\begin{aligned}
		& \sum_{k \in \mathcal{K}} \sum_{s\in \mathcal{F}^k \cup \{0\}}\lambda_{s}^k r_{ij}^{k,s} \leq C_{ij}, ~  \forall~(i,j) \in \mathcal{L}, \\
		& \!\!\!\!\sum_{j: (j,i) \in \mathcal{{L}}} r_{ji}^{k,s} - \sum_{j: (i,j) \in \mathcal{{L}}} r_{ij}^{k,s} =b_i^{k,s}(\bar{\x}),\\
		& \qquad \qquad\qquad\quad \forall~ i \in \mathcal{I},~\forall~s \in \mathcal{F}^k\cup \{0\},~\forall~k \in \mathcal{K},\\
		& r_{ij}^{k,s} \geq 0,~\forall~(i,j)\in \mathcal{L},~\forall~s\in \F^k\cup \{0\},~\forall~k\in \mathcal{K},
	\end{aligned}
\end{equation}
has a feasible solution $\bar{\r}$. 
If problem \eqref{tr} has a feasible solution $\bar{\r}$, then $(\bar{\x}, \bar{\y}, \bar{\r})$ is an optimal solution of problem \eqref{ns}; otherwise, $(\bar{\x}, \bar{\y})$ cannot define a feasible solution of problem \eqref{ns}
and needs to be refined based on the information feedback from solving problem \eqref{tr}.

\subsubsection{Information feedback}
To enable the information feedback between the \FP and \TR subproblems and develop the decomposition algorithm for solving problem \eqref{ns}, we need the following Farkas' Lemma \cite[Chapter 3.2]{Conforti2014}.
\begin{lemma}[Farkas' Lemma]\label{fakas}
	The linear system $A\r \leq \bb$, $C\r=\d$, and $\r \geq 0$ has a feasible solution $\r$ if and only if $\a^\top \bb + \b^\top \d  \geq 0$ for all $(\a,\b)$ satisfying $\a^\top A + \b^\top C \geq 0$ and $\a \geq 0$. 
\end{lemma}

In order to apply Farkas' Lemma,  we define dual variables $\alpha_{ij} \geq 0$ and $\beta_i^{k,s}$ associated with first and second families of linear inequalities in problem \eqref{tr}, and obtain the dual problem of the above \TR subproblem as follows:
	\begin{align} \revv{\minimize_{\substack{\a,\,\b}}}&  
		\sum_{(i,j) \in \CL} C_{ij} \alpha_{ij}  + \sum_{i \in \I} \sum_{k \in \K} \sum_{s \in \F^k\cup\{0\}} b_i^{k,s} (\bar{\x}) \beta_i^{k,s}  \nonumber\\
		{\revv{\text{subject to}}~} & \lambda_s^k \alpha_{ij} - \beta_i^{k,s} + \beta_j^{k,s}\geq 0,\nonumber\\
		&  \quad ~\forall~(i,j)\in \mathcal{L},~\forall~s\in \F^k\cup \{0\},~\forall~k\in \mathcal{K},\nonumber\\
		& \alpha_{ij} \geq 0, ~\forall~(i,j) \in \mathcal{L}.	\label{dtr}
		\tag{TR-D}
	\end{align}
Then problem \eqref{tr} has a feasible solution if and only if the optimal value of its dual \eqref{dtr} is nonnegative.

\begin{algorithm}[t]
	\caption{The decomposition algorithm for problem \eqref{ns}}
	\label{alg1}
	\begin{algorithmic}[1]
		\renewcommand{\algorithmicrequire}{\textbf{Input:}}
		\renewcommand{\algorithmicensure}{\textbf{Output:}}
		\STATE Set $\mathcal{C}\leftarrow \emptyset$ and $t \leftarrow 1$;
		\WHILE {$t \leq \IterMax $}
		\STATE Add the inequalities in set $\mathcal{C}$ into problem \eqref{vnf} to obtain a tightened \FP subproblem;
		\STATE If the tightened problem is feasible, let $(\bar{\x}, \bar{\y})$ be its solution; otherwise, stop and declare that problem \eqref{ns} is infeasible;	
		\STATE If problem \eqref{dtr} is unbounded, let $(\bar{\a}, \bar{\b})$ be its extreme ray such that \eqref{infinitecase} holds; otherwise, stop and declare that {$\bar{\x}$} {is} an optimal solution of problem \eqref{ns};
		\STATE Add the valid inequality \eqref{BDFcut} into set $\mathcal{C}$;
		\STATE $t \leftarrow t+1$;
		\ENDWHILE
	\end{algorithmic} 
\end{algorithm}

For the above problem \eqref{dtr}, as the all-zero vector is feasible, one of the following two cases must happen: (i) its optimal value is equal to zero; (ii) it is unbounded.
In case (i), by \cref{fakas}, problem \eqref{tr} has a feasible solution $\bar{\r}$, and
point $(\bar{\x}, \bar{\y}, \bar{\r})$ must be an optimal solution of problem \eqref{ns}.
In case (ii), there exists an extreme ray $(\bar{\a}, \bar{\b})$ of problem \eqref{dtr} for which 
\begin{equation}
	\label{infinitecase}
	\sum_{(i,j) \in \CL} C_{ij} \bar{\alpha}_{ij} + \sum_{i \in \I} \sum_{k \in \K} \sum_{s \in \F^k\cup \{0\}} b_i^{k,s} (\bar{\x}) \bar{\beta}_{i}^{k,s} < 0,
\end{equation}
meaning that point $(\bar{\x}, \bar{\y})$ cannot define a feasible solution of problem \eqref{ns}.
Notice that the extreme ray $(\bar{\a}, \bar{\b})$ satisfying  \eqref{infinitecase} can be obtained if problem \eqref{dtr} is solved by \revv{the} primal or dual simplex algorithm; see \cite[Chapter 3]{Dantzig1997}.
To remove point $(\bar{\x}, \bar{\y})$ from problem \eqref{vnf}, we can construct the following linear inequality
\begin{equation}
	\label{BDFcut}
	\sum_{(i,j) \in \CL} \bar{\alpha}_{ij}  C_{ij} + \sum_{i \in \I} \sum_{k \in \K} \sum_{s \in \F^k\cup \{0\}}\bar{\beta}_{i}^{k,s}  b_i^{k,s} ({\x})  \geq 0.
\end{equation}

\rev{Observe that
	\begingroup
	\allowdisplaybreaks
	\begin{align}
			& \sum_{i \in \I} \sum_{k \in \K} \sum_{s \in \F^k\cup \{0\}}\bar{\beta}_{i}^{k,s}  \left( \sum_{j: (j,i) \in \mathcal{{L}}} r_{ji}^{k,s} - \sum_{j: (i,j) \in \mathcal{{L}}} r_{ij}^{k,s} \right)\nonumber\\
			& ~~~=\sum_{k \in \K} \sum_{s \in \F^k\cup \{0\}} \sum_{i \in \I} \bar{\beta}_{i}^{k,s}  \left( \sum_{j: (j,i) \in \mathcal{{L}}} r_{ji}^{k,s} - \sum_{j: (i,j) \in \mathcal{{L}}} r_{ij}^{k,s} \right)\nonumber\\
			& ~~~\revv{\stackrel{(a)}{=}\sum_{k \in \K} \sum_{s \in \F^k\cup \{0\}} \left( \sum_{(i,j) \in \mathcal{{L}}} \bar{\beta}_{j}^{k,s}   r_{ij}^{k,s}-\sum_{(i,j) \in \mathcal{{L}}} \bar{\beta}_{i}^{k,s}   r_{ij}^{k,s}\right)}\nonumber\\
			& ~~~\revv{=\sum_{k \in \K} \sum_{s \in \F^k\cup \{0\}} \sum_{(i,j) \in \mathcal{{L}}}\left(  \bar{\beta}_{j}^{k,s}   - \bar{\beta}_{i}^{k,s} \right)r_{ij}^{k,s}},	\label{tmpeq}
	\end{align}
	\endgroup
	where (a) follows from
	\begin{align*}
		& \sum_{i \in \I}\bar{\beta}_{i}^{k,s}   \sum_{j: (j,i) \in \mathcal{{L}}} r_{ji}^{k,s} = \sum_{(j,i) \in \mathcal{{L}}} \bar{\beta}_{i}^{k,s}   r_{ji}^{k,s} = \sum_{(i,j) \in \mathcal{{L}}} \bar{\beta}_{j}^{k,s}   r_{ij}^{k,s} , \\
		& -\sum_{i \in \I}\bar{\beta}_{i}^{k,s}   \sum_{j: (i,j) \in \mathcal{{L}}} r_{ij}^{k,s} = -\sum_{(i,j) \in \mathcal{{L}}} \bar{\beta}_{i}^{k,s}   r_{ij}^{k,s}.
	\end{align*}
}%
\rev{Hence,}
\begin{align*}
	& \sum_{(i,j) \in \CL}  \bar{\alpha}_{ij}  C_{ij}+ \sum_{i \in \I} \sum_{k \in \K} \sum_{s \in \F^k\cup \{0\}}\bar{\beta}_{i}^{k,s}  b_i^{k,s} ({\x}) \qquad \qquad\qquad\qquad\qquad \\
	&~~~\rev{\stackrel{(a)}{\geq}} \sum_{(i,j) \in \CL}   \bar{\alpha}_{ij}\left( \sum_{k \in \mathcal{K}} \sum_{s\in \mathcal{F}^k \cup \{0\}}\lambda_{s}^k r_{ij}^{k,s}\right) + \\ 
	& \quad \qquad\sum_{i \in \I} \sum_{k \in \K} \sum_{s \in \F^k\cup \{0\}}\bar{\beta}_{i}^{k,s}  \left( \sum_{j: (j,i) \in \mathcal{{L}}} r_{ji}^{k,s} - \sum_{j: (i,j) \in \mathcal{{L}}} r_{ij}^{k,s} \right)\\
	& ~~~\rev{\stackrel{(b)}{=}}\sum_{(i,j) \in \CL}   \sum_{k \in \mathcal{K}} \sum_{s\in \mathcal{F}^k \cup \{0\}} \left(\lambda_{s}^k \bar{\alpha}_{ij} - \bar{\beta}_{i}^{k,s} + \bar{\beta}_{j}^{k,s}\right)  r_{ij}^{k,s} \rev{\stackrel{(c)}{\geq}} 0,
\end{align*}
where \rev{(a)} follows from $\bar{\alpha}_{ij}\geq 0$, \eqref{linkcapcons}, and \eqref{SFC}\rev{;} \rev{(b) follows from \eqref{tmpeq}}; and \rev{(c)} follows from \eqref{rbounds} and the first set of inequalities in \eqref{dtr}.
\rev{This shows that inequality \eqref{BDFcut} is a valid inequality for problem \eqref{ns} in the sense that \eqref{BDFcut} holds at every feasible solution of problem \eqref{ns}.}
\rev{Therefore,} we can add \eqref{BDFcut} into problem \eqref{vnf} to obtain a tightened \FP problem, which is still a relaxation of problem \eqref{ns}.
Moreover, by \eqref{infinitecase}, the current solution $(\bar{\x}, \bar{\y})$ is infeasible to this tightened \FP problem.
As a result, we can solve the tightened \FP problem again to obtain a new solution.
This process is repeated until case (i) happens.
The above procedure is summarized as \cref{alg1}.

\subsection{Analysis Results and Remarks}\label{subsec:analysis}
In this subsection, we provide some analysis results and remarks on the proposed decomposition algorithm. 

First, at each iteration of \cref{alg1}, the inequality in \eqref{BDFcut} will remove a feasible point $\bar{\x}$ from problem \eqref{vnf} (if it is an infeasible point of problem \eqref{ns}). 
This, together with the fact that the feasible binary solutions $\bar{\x}$ of problem \eqref{vnf} are finite, implies \rev{that, for problem \eqref{ns},  \cref{alg1} with a sufficiently large $\IterMax$ will either find one of its optimal solutions (if the problem is feasible) or declare the infeasibility (if it is infeasible).}
Although the worst-case iteration complexity of the proposed decomposition algorithm grows exponentially with the number of variables, our simulation results show that it usually finds an optimal solution of problem \eqref{ns} within a small number of iterations; see \cref{sect:numres} further ahead.

Second, each iteration of \cref{alg1} needs to solve an \MBLP problem \eqref{vnf} and an LP problem \eqref{dtr}.
It is worthwhile highlighting that these two subproblems are much  easier to solve than the original problem \eqref{ns} and this feature makes the proposed decomposition algorithm particularly \rev{suitable to solve} large-scale \NS problems.
To be specific, although subproblem \eqref{vnf} is still an MBLP problem, \rev{both of its numbers of variables and constraints are}  $\mathcal{O}(|\mathcal{V}|\sum_{k\in \mathcal{K}}\ell_k)$.
\rev{This is significantly smaller than those of problem \eqref{ns},
	which are 
	$$\mathcal{O}\left((|\mathcal{V}|+|\mathcal{L}|)\sum_{k\in \mathcal{K}}\ell_k\right)~\text{and}~ \mathcal{O}\left({\min}\left\{|\mathcal{I}|\sum_{k\in \mathcal{K}}\ell_k, |\mathcal{L}|\right\}\right),$$ 
	respectively, especially when the number of cloud nodes is much smaller than the numbers of nodes and links in the network.}
Subproblem \eqref{dtr} is an LP problem, which is polynomial time solvable within the complexity $\mathcal{O}\left (\left(|\mathcal{L}|\sum_{k\in \mathcal{K}}\ell_k\right)^{3.5}\right)$ \cite[\rev{Chapter} 6.6.1]{Ben-Tal2001}.

\revv{Finally}, as compared with traditional two-stage heuristic algorithms \cite{Chowdhury2012}-\cite{Luizelli2015} that solve subproblems \eqref{vnf} and \eqref{dtr} in a one-shot fashion, our proposed decomposition algorithm iteratively solves the two subproblems enabling a useful information feedback between them (as shown in \cref{ideaTSD}), which enables it to return an optimal solution of problem \eqref{ns} and thus much better solutions than those in \cite{Chowdhury2012}-\cite{Luizelli2015}.

%% file: section_ineq.tex
\section{Valid Inequalities}
\label{sect:ineq}

As seen in \cref{sect:TSD}, problem \eqref{vnf} is a relaxation of problem \eqref{ns} 
in which all traffic routing related variables and constraints (i.e., \eqref{linkcapcons}, \eqref{SFC}, and \eqref{rbounds}) are {completely} dropped. 
The goal of this section is to derive some valid inequalities from the dropped constraints in order to strengthen problem \eqref{vnf} and thus accelerate the convergence of \cref{alg1}.
In particular, we shall derive two families of valid inequalities from problem \eqref{ns} in the $(\x, \y)$-space 
by explicitly taking the connectivity of the nodes and the limited link capacity constraints in the underlying network into consideration, and add them into problem \eqref{vnf}, thereby making the convergence of Algorithm \ref{alg1} more quickly.

\subsection{Connectivity based Inequalities}
\label{susbec:connectivityIneq}
The first family of valid inequalities is derived by considering the connectivity of the source node $S^k$ and the node hosting function $f_{1}^k$, the node hosting function $f_{\ell_k}^k$ and the destination $D^k$, and the two nodes hosting two adjacent functions in the SFC of service $k$. 
To present the inequalities, we define  the following notations 
\begin{align}
	& \!\!\V(S^k) \!=\! \left\{ v \in \V\mid\text{there is no path from $S^k$ to $v$} \right\},\label{VSKdef}\\
	& \!\!\V(D^k) \!=\! \left\{ v \in \V\mid \text{there is no path from $v$ to $D^k$} \right\},\label{VDKdef}\\
	& \!\!\V(v_0) \!=\! \left\{ v \in \V\backslash \{v_0\}\mid \text{there is no path from $v_0$ to $v$} \right\},	\label{v0def}
\end{align}
where $v_0 \in \V$.
The computations of subsets $\V(S^k)$ and $\V(D^k)$ in \eqref{VSKdef} and \eqref{VDKdef}, respectively, for all $k \in \K$, and $\V(v_0)$ in \eqref{v0def} for all $v_0 \in \V$,
can be done by computing the {transitive closure matrix} $M^G \in \{0,1\}^{|\I| \times |\I|}$ of the directed graph $G$ defined by 
\begin{equation*}
	M^G_{ij} = \left\{\begin{array}{ll}
		1, & \text{if there exists a path from node $i$ to node $j$};\\[5pt]
		0, & \text{otherwise}.
	\end{array}\right.
\end{equation*}
The transitive closure matrix $M^G$ can be computed by applying the	 breadth first search for each node $i \in \I$ with a complexity  of $\mathcal{O}(|\I|(|\I|+|\CL|))$.

\subsubsection{Three forms of derived inequalities}
We first present the simplest form of connectivity based inequalities to ease the understanding. 
To ensure the connectivity from the source node $S^k$ to the node hosting function $f_1^k$, function $f_{1}^k$ cannot be hosted at cloud nodes $v \in \V(S^k)$:
\begin{equation}
	\label{SFCineq1} 
	x^{k,1}_v = 0, ~\forall~v \in \V(S^k),~\forall~ k \in \K.
\end{equation}
Similarly, to ensure the connectivity from the node hosting function $f_{\ell_k}^k$ to the destination node $D^k$,  function $f_{\ell_k}^k$ cannot be hosted at cloud nodes $v \in \V(D^k)$:
\begin{equation}
	x_{v}^{k,\ell_k} = 0,~ \forall~v \in \V(D^k), ~\forall~ k \in \K. \label{SFCineq2}
\end{equation}
Suppose that function $f_{s}^k$ is hosted at cloud node $v_0 \in \mathcal{V}$.
To ensure the connectivity from node $v_0$ to the node hosting function  $f_{s+1}^k$, function  $f_{s+1}^k$ cannot be hosted at cloud node $v \in \V(v_0)$.
This constraint can be written by 
\begin{align}
	& x_{v_0}^{k,s}	+ x^{k,s+1}_v \leq 1,\nonumber\\
	&  \qquad \quad~\forall~v \in \V(v_0),~\forall~v_0 \in \V,~\forall~s \in \F^k,~\forall~k \in \mathcal{K}. \label{SFCineq3} 
\end{align}

Inequalities \eqref{SFCineq1}--\eqref{SFCineq3} can be further extended. 
Indeed, for each $k \in \K$ and $s \in \F^k$, there must be a path from the source node $S^k$ to the node hosting function $f_s^k$, and the same to the node hosting function $f_s^k$ and the destination node $D^k$. 
This implies that function $f_s^k$ cannot be hosted at cloud node $v \in \V(S^k)\cup \V(D^k)$, and as a result, inequalities \eqref{SFCineq1} and \eqref{SFCineq2} can be extended to
\begin{equation}
	\label{SFCineq-1}
	x^{k,s}_v = 0,~\forall~v \in \V(S^k) \cup \V(D^k),~\forall~s \in \F^k,~\forall~ k \in \K.
\end{equation}
Similarly, for $s \in \F^k \backslash \{\ell_k\}$, there must be a path from the node hosting function $f_{s}^k$  to the node hosting function $f_{s_0}^k$, $s_0 \in \{ s+1, \ldots, \ell_k \}$, implying that
\begin{align}
	& x_{v_0}^{k,s}	+ x_v^{k,s_0} \leq 1, ~\forall~v \in \V(v_0),~\forall~v_0 \in \V, \nonumber \\
	&\quad\forall~s_0 \in \{ s+1, \ldots, \ell_k\},~\forall~\revv{s \in \F^k\backslash \{\ell_k\}},~\forall~k \in \mathcal{K}.\label{SFCineq-2}
\end{align}
Clearly, inequalities \eqref{SFCineq-1}--\eqref{SFCineq-2} include inequalities \eqref{SFCineq1}--\eqref{SFCineq3} as special cases. 

Next, we develop a more compact way to present the connectivity based inequalities \eqref{SFCineq-2}.
To proceed, we provide the following proposition, which is useful in deriving the new inequalities and studying the relations between different inequalities in  \cref{thm1}.
\begin{proposition}
	\label{fact}
	The following three statements are true:
	\begin{itemize}
		\item [(i)] for any given $v_0 \in \V$, there does not exist a path from any node in $\V \backslash \V(v_0)$ to any node in $ \V(v_0)$ and $\V(v_0) \subseteq \V(v')$ holds for all $v' \in \V\backslash \V(v_0)$;
		\item [(ii)] for any given $k \in \K$, there does not exist a path from any node in $\V \backslash \V(S^k)$ to any node in $ \V(S^k)$ and $\V(S^k) \subseteq \V(v_0)$ holds for all $v_0 \in \V\backslash \V(S^k)$; and 
		\item [(iii)] for any given $k \in \K$, there does not exist a path from any node in $\V(D^k)$ to any node in $\V \backslash \V(D^k)$ and $\V\backslash \V(D^k)\subseteq \V(v_0) $ holds for all $v_0 \in \V(D^k)$.
	\end{itemize}
\end{proposition}
\begin{proof}
	Let $v_1 \in \V\backslash \V(v_0)$ and $v_2 \in \V(v_0)$.
	Clearly, the statement is true if $v_1 = v_0$.
	Otherwise, suppose that there exists a path from node $v_1$ to node $v_2$. 
	By the definition of $\V(v_0)$ in \eqref{v0def}, there must exist a path from node $v_0$ to node $v_2$, i.e., $v_0 \rightarrow v_1 \rightarrow v_2$; see \cref{fignodes} for an illustration. 
	However, this contradicts with $v_2 \in \V(v_0)$. 
	As a result, no path exists from node $v_1$ to node $v_2$, which shows the first part of (i). 
	To prove the second part of (i), we note that there does not exist a path from $v' \in \V\backslash \V(v_0)$ to any node in $\V(v_0)$. 
	Moreover, according to the definition of $\V(v')$ in \eqref{v0def}, $\V(v')$ is the set of all nodes to which there does not exist a path from $v'$.
	As a result, $\V(v_0) \subseteq \V(v')$.
	The proofs of (ii) and (iii) are similar.
\end{proof}

\begin{figure}
	\begin{center}
\begin{tikzpicture}[> = Stealth]
	\node (1) at (0,0) {};\filldraw(1.east) circle (2pt);
	\node (2) [right of = 1] {};\filldraw(2.east) circle (2pt);
	\node (3) [right of = 2,label=above:$~~~v_1$] {};\filldraw(3.east) circle (2pt);
	\node (4) [above of = 2,label=above:$v_0$] {};\filldraw(4.east) circle (2pt);
	\node[fit = (1) (3) (4) , ellipse, draw=blue, thick, minimum width=2cm, label=below:$\V\backslash\V(v_0)$] (A){};
	
	\node (1_1)  at (4,0)[label=above:$v_2$] {};\filldraw(1_1.east) circle (2pt);
	\node (2_2) [right of=1_1] {};\filldraw(2_2.east) circle (2pt);
	\node (3_3) [right of=2_2] {};\filldraw(3_3.east) circle (2pt);
	\node (4_3) [above of=2_2] {};
	\node[fit = (1_1) (3_3) (4_3), ellipse, draw=black, thick, minimum width=2cm, label=below:\(\V(v_0)\)] (B) {};
	\draw[->, shorten >=.1cm,thick] (4.east) to (1.east);
	\draw[->, shorten >=.1cm,thick] (4.east) to (2.east);
	\draw[->, shorten >=.1cm,thick] (4.east) to (3.east);
	
	\draw[->, shorten >=.1cm,color=red,thick,dash dot] (3.east) to (1_1.east);
	
\end{tikzpicture}
\end{center}
\caption{A toy example for illustrating that (i) there does not exist a path from any node in $\V\backslash\V(v_0)$ to any node in $\V(v_0)$; and (ii) for all $v \in \V\backslash(\V(v_0)\cup\{v_0\})$, there exists a path from node $v_0$ to node $v$.
The red dot line is drawn for the contradiction proof of \cref{fact}.}
\label{fignodes}
\end{figure}
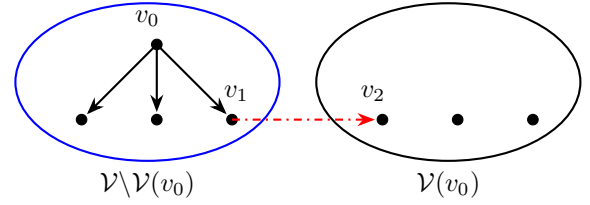
\noindent The result in \cref{fact} (i) implies that if function $f^k_{s}$ is hosted at a node in $\V \backslash \V(v_0)$, then function $f_{s+1}^k$ must also be hosted at a node in $\V \backslash \V(v_0)$.
This, together with \eqref{onlyonenode}, implies that 
\begin{align}
	& \sum_{v \in \V \backslash \V(v_0)} x^{k,s}_v \leq \sum_{v \in \V \backslash \V(v_0)} x^{k,s+1}_v,\nonumber \\
	& \qquad \qquad\qquad \forall~v_0 \in \V,~\forall~s\in \F^k\backslash\{\ell_k\}, ~\forall~k \in \K.\label{SFCineq-3}
\end{align}
It should be mentioned that if $\V(v_0)= \emptyset$, then it follows from \eqref{onlyonenode} that both the left-hand and right-hand sides of \eqref{SFCineq-3} are all equal to one, and thus \eqref{SFCineq-3} holds naturally.

\subsubsection{Relations of the inequalities}
Now we study the relations of adding Eqs. \eqref{SFCineq1}--\eqref{SFCineq3}, Eqs. \eqref{SFCineq-1}--\eqref{SFCineq-2}, and Eqs. \eqref{SFCineq-1} and \eqref{SFCineq-3} into problem \eqref{vnf}.
To do this, let
\begin{align*}
	& \X^1 = \{ \x \mid \eqref{onlyonenode},~ \eqref{xbounds},~\eqref{SFCineq1}\text{--}\eqref{SFCineq3} \}, \\
	& \X^2 = \{ \x \mid \eqref{onlyonenode},~\eqref{xbounds},~\eqref{SFCineq-1}\text{--}\eqref{SFCineq-2} \}, \\
	& \X^3 = \{ \x \mid \eqref{onlyonenode},~\eqref{xbounds},~\eqref{SFCineq-1},~ \eqref{SFCineq-3} \}.
\end{align*}
Moreover, let $\X^a_{\LIN}$ ($a=1, 2,3$) denote the linear relaxation of $\X^a$ in which \eqref{xbounds} is replaced by 
\begin{equation}
	x^{k,s}_v \in [0,1],~\forall~v\in\mathcal{{V}},~\forall~s\in \mathcal{F}^k,~\forall~k\in\mathcal{K}.\label{xboundsL}
\end{equation}
The following theorem clearly characterizes the relations of three sets $\X^1$, $\X^2$, and $\X^3$ as well as their linear relaxations. 

\begin{theorem}
	\label{thm1}
	The following two relationships are true: (i) $\X^1 = \X^2 = \X^3$; and (ii) $\X^3_{\LIN} \subseteq \X^2_{\LIN} \subseteq \X^1_{\LIN}$ \footnote{\rev{Note that in general, it is possible that two 0-1 sets, represented by different inequalities and binary constraints $\x \in \{0,1\}^n$, are equivalent but their linear relaxations are different; see \cite[Chapter 1.6--1.7]{Wolsey2021}.}}.
\end{theorem}
\begin{proof}
	The proof can be found in Appendix \ref{appendixProof}.
\end{proof}

The result in \cref{thm1} (i) implies that the feasible regions of the newly obtained problems by adding Eqs. \eqref{SFCineq1}--\eqref{SFCineq3}, Eqs. \eqref{SFCineq-1}--\eqref{SFCineq-2}, or Eqs. \eqref{SFCineq-1} and \eqref{SFCineq-3} into problem \eqref{vnf} are the same.
However, compared with the other two, adding Eqs. \eqref{SFCineq-1} and \eqref{SFCineq-3} into problem \eqref{vnf} enjoys the following two advantages.
First, adding \eqref{SFCineq-1} and \eqref{SFCineq-3} into problem \eqref{vnf} yields a much more compact problem formulation:
\begin{equation}
	\label{vnf1}
	\tag{FP-I}
	\begin{aligned} \revv{\minimize_{\substack{\x,\,\y}}}~&  \sum_{v \in \mathcal{V}}p_v y_v + \sum_{v\in\mathcal{V}} \sum_{k \in \mathcal{K}} \sum_{s \in \F^k} c^{k,s}_v x^{k,s}_v  \\
		{\revv{\text{subject to}}~}&  \eqref{onlyonenode}\text{--}\eqref{nodecapcons},~\eqref{ybounds},~\eqref{xbounds},~ \eqref{SFCineq-1},~ \eqref{SFCineq-3},
	\end{aligned}
\end{equation}
as compared with the other two alternatives.
Indeed, 
\begin{itemize}
	\item compared with  \eqref{SFCineq1}--\eqref{SFCineq2}, more variables can be fixed to $0$ by \eqref{SFCineq-1}, and thus can be removed from problem \eqref{vnf1};
	\item compared with those in \eqref{SFCineq3} or \eqref{SFCineq-2}, the number of inequalities in \eqref{SFCineq-3} is potentially much more smaller.
	In particular, the number of inequalities in \eqref{SFCineq-3} is $\mathcal{O}(|\V|\sum_{k \in \K} \ell_k)$ while those in  \eqref{SFCineq3} and \eqref{SFCineq-2} are 
	$\mathcal{O}(|\mathcal{V}|^2\sum_{k \in \mathcal{K}}\ell_k)$ and $\mathcal{O}(|\mathcal{V}|^2\sum_{k \in \mathcal{K}}\ell_k^2)$, respectively.
\end{itemize}
Second, \cref{thm1} (ii) shows that adding \eqref{SFCineq-1} and \eqref{SFCineq-3} into problem \eqref{vnf} {can provide a stronger \LP bound}, as compared with the other two alternatives.
An  example for illustrating this is provided in the following \cref{example1}.
The strong LP relaxations and corresponding LP bounds play a crucial role in efficiently solving problem \eqref{vnf} by calling a standard \MBLP solver; see \cite[Section 2.2]{Conforti2014}.
Due to the above two reasons, we conclude that solving problem \eqref{vnf} with inequalities \eqref{SFCineq-1} and \eqref{SFCineq-3} (i.e., problem \eqref{vnf1}) is much more computationally efficient than solving the other two alternatives.

\begin{example}
	\label{example1}
	Consider the toy example in Fig. \ref{example}.
	Nodes $1$, $2$, and $3$ are cloud nodes whose computational capacities are $\mu_1 = +\infty$, $\mu_2 = +\infty$, and $\mu_3= 3$. 
	The power consumptions of the three cloud nodes are $p_1=p_2=p_3 = 0$. 
	There are two different functions, $f^1$ and $f^2$, and all three cloud nodes can process the two functions. 
	The power consumptions of placing functions to cloud nodes are $0$ except that the power consumptions of placing function $f^1$ to cloud nodes $1$ and $2$ are $1$.

	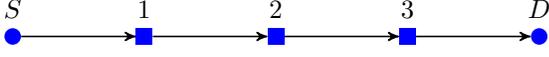
\begin{figure}[t]
		\begin{center}
			\begin{tikzpicture}[->,>=stealth',shorten >=0pt,auto,node distance=7cm,
				semithick]
				\tikzstyle{every state}=[fill=blue,draw=none,scale=0.25,text=white]	
				\node[state,label=above:{$S$}] (A)                  {};
				\node[state,label=above:{$1$},rectangle]         (B) [right of=A] {};
				\node[state,label=above:{$2$},rectangle]         (C) [right of=B] {};
				\node[state,label=above:{$3$},rectangle]         (D) [right of=C] {};
				\node[state,label=above:{$D$}]         (E) [right of=D] {};
				
				\path
				(A) edge             node {} (B)
				(B) edge 			 node {} (C)
				(C) edge 			 node {} (D)
				(D) edge 			 node {} (E)
				;
			\end{tikzpicture}
			\caption{A toy network example where the rectangular nodes are cloud nodes and the circular nodes are the source and destination of a service.}
			\label{example}
		\end{center}
	\end{figure}
	Suppose that there is a flow with its source and destination nodes being $S$ and $D$, respectively, its SFC being $f^1 \rightarrow f^2$, and its data rate being $2$.
	Then, solving the LP relaxation of problem \eqref{vnf1} will return a solution with the objective value being $\frac{1}{4}$. 
	However, solving the \LP relaxation of problem \eqref{vnf} with constraints \eqref{SFCineq1}--\eqref{SFCineq3} or with constraints \eqref{SFCineq-1}--\eqref{SFCineq-2} returns a solution with the objective value being $\frac{1}{6}$; see Appendix \ref{appendixA} for more details.
	This example clearly shows that adding constraints \eqref{SFCineq-1} and \eqref{SFCineq-3} into problem \eqref{vnf} can provide a strictly stronger \LP bound than adding constraints \eqref{SFCineq1}--\eqref{SFCineq3} or  \eqref{SFCineq-1}--\eqref{SFCineq-2} into problem \eqref{vnf}.
\end{example}

\subsubsection{Theoretical justification of the  connectivity based inequalities}
In this part, we present a theoretical analysis result of the connectivity based inequalities \eqref{SFCineq-1} and \eqref{SFCineq-3}
in the special case where the capacities of all links are  infinite.
We shall show that  problems \eqref{ns} and \eqref{vnf1} are equivalent in this special case in the sense that 
if $(\x,\y,\r)$ is a feasible solution of problem \eqref{ns}, then $(\x,\y)$ is a feasible solution of problem \eqref{vnf1};
and conversely, if $(\x,\y)$ is a feasible solution of problem \eqref{vnf1}, then there exists a vector $\r$ such that $(\x,\y,\r)$ is a feasible solution of problem \eqref{ns}.

\begin{theorem}
	\label{thm3}
	Suppose that $C_{ij} = + \infty$ for all $(i,j) \in \CL$.
	Then problem \eqref{ns} is equivalent to problem \eqref{vnf1}.
\end{theorem}
\begin{proof}
	As both \eqref{SFCineq-1} and \eqref{SFCineq-3} are valid inequalities, problem \eqref{vnf1} can be seen as a relaxation of problem \eqref{ns}. 
	Therefore,  if $(\x,\y,\r)$ is a feasible solution of problem \eqref{ns}, then $(\x,\y)$ must also be a feasible solution of problem \eqref{vnf1}.
	
	To prove the other direction, let $({\x}, {\y})$  be a feasible solution 
	of problem \eqref{vnf1}.
	By \cref{thm1}, Eqs. \eqref{SFCineq-1} and \eqref{SFCineq-3} imply Eqs. \eqref{SFCineq1}--\eqref{SFCineq3}. 
	Therefore, for the solution $({\x}, {\y})$,
	(i) by \eqref{SFCineq1}, there exists a path from the source $S^k$ to the cloud node hosting function $f_1^k$;
	(ii) by \eqref{SFCineq2}, there exists a path from the cloud node hosting function $f^k_{\ell_k}$ to the destination $D^k$; and
	(iii) by \eqref{SFCineq3}, there exists a path from the cloud node hosting  function $f_{s}^k$ to the cloud node hosting function $f_{s+1}^k$ for all $s \in \F^k\backslash\{\ell_k\}$.
	{The above,} together with the assumption that $C_{ij} = + \infty$ for all $(i,j) \in \CL$, implies that problem \eqref{tr} has a feasible solution $\r$, and thus $({\x}, {\y}, \r)$ is feasible solution of problem \eqref{ns}.
\end{proof}

\cref{thm3} justifies the effectiveness of the derived connectivity based inequalities \eqref{SFCineq-1} and \eqref{SFCineq-3} for the proposed \rev{\cref{alg1}}.
Specifically, if all links' capacities in the underlying network are infinite, then equipped with the connectivity based inequalities \eqref{SFCineq-1} and \eqref{SFCineq-3}, the proposed \rev{\cref{alg1}}  terminates in a single iteration.
Even when some or all links' capacities are finite, the connectivity based inequalities \eqref{SFCineq-1} and \eqref{SFCineq-3} can still accelerate the convergence of \cref{alg1}, as they ensure that,
for a solution $(\x, \y)$ returned by solving \eqref{vnf1},
necessary paths exist from source nodes to the nodes hosting functions in the SFC to destination nodes for all services.

\subsection{Link-Capacity based Inequalities}

In this subsection, we derive valid inequalities by taking the (possibly) limited link capacity of the network into account to further strengthen problem \eqref{vnf1}.
As enforced in constraints \eqref{nodecapcons}, the total amount of functions processed at cloud node $v$ is limited by its capacity $\mu_v$, that is, 
the total data rate of these functions cannot exceed $\mu_v$.
In addition to the above, the total amount of functions processed at cloud node $v$ is also limited by the overall capacity of its incoming links and outgoing links.
To be specific, (i) if function $f_{s+1}^{k}$ is processed at cloud node $v$ but function $f_s^k$ is not processed at cloud node $v$, then the total data rate of flow $(k,s)$ on node $v$'s incoming links should be $\lambda_s^k$;
	and the sum of the total data rate of all the above cases should not exceed the overall capacity of node $v$'s incoming links.
(ii) if function $f_s^k$ is processed at cloud node $v$ but function $f_{s+1}^{k}$ is not processed at cloud node $v$, then the total data rate of flow $(k,s)$ on node $v$'s outgoing links should be $\lambda_s^{k}$;
	and the sum of the total data rate of all the above cases should not exceed the overall capacity of node $v$'s outgoing links.
The link-capacity based valid inequalities in this subsection are derived based on the above observations.

\subsubsection{Derived inequalities}
We first consider case (i) in the above.
Observe that all constraints in \eqref{SFC} associated with cloud node $v\in \V$ can be written as
\begin{align}
		& \sum_{j: (j,v) \in \mathcal{{L}}} r_{jv}^{k,s} - \sum_{j: (v,j) \in \mathcal{{L}}} r_{vj}^{k,s} = b_v^{k,s}(\x),\nonumber \\
		&\qquad \qquad\qquad\qquad\qquad\quad ~\forall~s \in \F^k\cup \{0\},~\forall~k \in \K,\label{SFCv}
\end{align}
where 
\begin{equation*}
	b_v^{k,s}(\x) =\left\{
	 \begin{array}{ll}
	 	 x^{k,s+1}_v, & \text{if}~s=0;\\[3pt]
	 	 x^{k,s+1}_v - x^{k,s}_v,&   \text{if}~1 \leq s < \ell_k;\\[3pt]
	 	 -x^{k,s}_{v},~&  \text{if}~s = \ell_k.
	 \end{array}\right.
\end{equation*}
By \eqref{rbounds}, we have $r_{jv}^{k,s}\geq 0$ for all $j$ with $(j,v) \in \CL$ and $r_{vj}^{k,s}\geq 0$ for all $j$ with $(v,j) \in \CL$, which, together with \eqref{SFCv}, implies
\begin{equation}
	\label{tineq1}
	\!\!\!\!\sum_{j :  (j,v)\in \mathcal{L}} r^{k,s}_{jv} \geq \left\{
	\begin{array}{ll}
		x^{k,s+1}_v, & \text{if}~s=0;\\[3pt]
		{(x^{k,s+1}_v - x^{k,s}_v)}^+, &  \text{if}~1 \leq s < \ell_k.
	\end{array}
	\right. 
\end{equation}
Here $(a)^+$ is equal to $a$ if $a > 0$ and $0$ otherwise.
Multiplying \eqref{tineq1} by $\lambda_{s}^k$ and summing them up for all $k \in \mathcal{K}$ and $s\in \F^k\backslash\{\ell_k\}\cup \{0\}$, we obtain 
\rev{
	\begin{equation*}
		\begin{aligned}
			&\sum_{k \in \mathcal{K}} \left( \lambda_0^k x^{k,1}_v + \sum_{s \in \F^k \backslash\{\ell_k\}} \lambda_s^k {(x^{k,s+1}_v - x^{k,s}_v)}^+  \right)  \\
			&\qquad \quad~\stackrel{(a)}{\leq}\sum_{k \in \mathcal{K}} \sum_{s \in \F^k\cup \{0\}\backslash\{\ell_k\}} \lambda_s^k \left (\sum_{j :  (j,v)\in \mathcal{L}} r^{k,s}_{jv}\right) \\
			& \qquad \quad ~=\sum_{j : (j,v)\in \mathcal{L}}\left(\sum_{k \in \mathcal{K}} \sum_{s \in \F^k \cup \{0\}\backslash \{\ell_k\}} \lambda_s^k r^{k,s}_{jv} \right) \\
			&\qquad \quad~\stackrel{(b)}{\leq} 
			\sum_{j : (j,v)\in \mathcal{L}}\left(\sum_{k \in \mathcal{K}} \sum_{s \in \F^k \cup \{0\}} \lambda_s^k r^{k,s}_{jv} \right)
			\stackrel{(c)}{\leq}    \sum_{j  :  (j,v)\in \mathcal{L}} C_{jv},
		\end{aligned}
	\end{equation*}%
	where (a) follows from \eqref{tineq1}; (b) follows from $r^{k, \ell_k}_{jv} \geq 0$ for all $k \in \K$ and $j \in \I$ with $(j,v) \in \CL$; and (c) follows from \eqref{linkcapcons}.}
By \eqref{xyxelation}, if $y_v=0$, then $x_{v}^{k,s}=0$ for all $k \in \K$ and $s \in \F^k$, and thus 
the above inequality can be further strengthened as 
\begin{align}
	& \sum_{k \in \mathcal{K}} \left( \lambda_0^k x^{k,1}_v + \sum_{s \in \F^k\backslash\{\ell_k\}}  \lambda_s^k {(x^{k,s+1}_v - x^{k,s}_v)}^+  \right)\nonumber\\
	&  \qquad \qquad\qquad\qquad\qquad\leq  \sum_{j  : (j,v)\in \mathcal{L}} C_{jv} y_v,~\forall~v \in \mathcal{V}.	\label{validineq1}
\end{align}

Next, we consider case (ii). 
Similarly, from \eqref{SFCv}, we can get
\begin{equation}\label{tineq2}
		\!\!\!\! \sum_{j: (v,j) \in \mathcal{{L}}} r^{k,s}_{vj} \geq \left\{
	\begin{array}{ll}
		\!\!\! {(x^{k,s}_v-x^{k,s+1}_v )}^+, & 	\!\! \text{if}~1 \leq s < \ell_k;\\[3pt]
			\!\!\!x^{k,s}_v, & 	\!\!\text{if}~s=\ell_k.
	\end{array}
	\right. 
\end{equation}
Using the same argument as deriving \eqref{validineq1}, we can obtain
\begin{align}
	&\sum_{k \in \mathcal{K}} \left( \lambda_{\ell_k}^k x_{v}^{k,\ell_k} + \sum_{s \in \F^k \backslash\{\ell_k\}} \lambda_s^k {(x^{k,s}_v - x^{k,s+1}_v)}^+  \right)\nonumber \\
	& \qquad\qquad\qquad\qquad\qquad  \leq   \sum_{j  :  (v,j)\in \mathcal{L}} C_{vj} y_v, ~\forall~ v\in \mathcal{V}.\label{validineq2}
\end{align}%
\rev{Adding inequalities \eqref{validineq1} and \eqref{validineq2} into \eqref{vnf1} yields a stronger relaxation of problem \eqref{ns}: 
	\begin{equation}
		\label{vnf2-1}
		\begin{aligned} {\minimize_{\substack{\revvv{\x,\,\y}}}}~&  \sum_{v \in \mathcal{V}}p_v y_v + \sum_{v\in\mathcal{V}} \sum_{k \in \mathcal{K}} \sum_{s \in \F^k} c^{k,s}_v x^{k,s}_v  \\
			{\revv{\text{subject to}}~}&  \eqref{onlyonenode}\text{--}\eqref{nodecapcons},~\eqref{ybounds},~\eqref{xbounds},~ \eqref{SFCineq-1},~ \eqref{SFCineq-3},~ \eqref{validineq1},~\eqref{validineq2}.
		\end{aligned}
	\end{equation}}
	
	\rev{Problem \eqref{vnf2-1} is a mixed integer nonlinear programming problem due to the nonlinear terms  ${(x^{k,s+1}_v - x^{k,s}_v)}^+$ and ${(x^{k,s}_v - x^{k,s+1}_v)}^+$ in \eqref{validineq1} and \eqref{validineq2}.
	However, they can be equivalently linearized.
	To do so, we first note that by $x_v^{k,s+1}, x_v^{k,s} \in \{0,1\}$,  ${(x^{k,s+1}_v - x^{k,s}_v)}^+$,  ${(x^{k,s}_v - x^{k,s+1}_v)}^+\in \{0,1\}$ must hold.
	Therefore, we can introduce binary variables
	\begin{align}
		& z^{k,s}_v \in \{0,1\},~\forall~v \in \mathcal{V},~\forall~s \in \F^k \backslash\{\ell_k\}, ~ \forall~k \in \mathcal{K}, \label{validineq11} \\
		& w^{k,s}_v \in \{0,1\},~\forall~v \in \mathcal{V},~\forall~s \in   \F^k \backslash\{\ell_k\}, ~\forall~k \in \mathcal{K}, \label{validineq21}
	\end{align}
	to denote 
	\begin{align}
		& z_v^{k,s}={(x^{k,s+1}_v - x^{k,s}_v)}^+,\,\forall\,v \in \mathcal{V},\,\forall\,s \in \F^k \backslash\{\ell_k\}, \, \forall\,k \in \mathcal{K}, \label{validineq10} \\
		& w_v^{k,s}={(x^{k,s}_v - x^{k,s+1}_v)}^+,\,\forall\,v \in \mathcal{V},\,\forall\,s \in   \F^k \backslash\{\ell_k\}, \,\forall\,k \in \mathcal{K}, \label{validineq20}
	\end{align}
	and rewrite \eqref{validineq1} and \eqref{validineq2} as the following constraints 
	\begin{align}
		& \sum_{k \in \mathcal{K}} \left( \lambda_0^k x^{k,1}_v + \sum_{s \in \F^k \backslash\{\ell_k\}} \lambda_s^k z^{k,s}_v  \right)\nonumber \\
		& \qquad\qquad\qquad\qquad\qquad \leq   \sum_{j  :  (j,v)\in \mathcal{L}} C_{jv} y_v,~\forall~ v\in \mathcal{V},\label{validineq5}\\
		& \sum_{k \in \mathcal{K}} \left( \lambda_{\ell_k}^k x^{k,\ell_k}_v + \sum_{s \in \F^k \backslash\{\ell_k\}} \lambda_s^k w^{k,s}_v  \right)\nonumber\\
		& \qquad\qquad\qquad\qquad\qquad \leq   \sum_{j  :  (v,j)\in \mathcal{L}} C_{vj} y_v, ~\forall~ v\in \mathcal{V}.\label{validineq6}
	\end{align}}%
	As a result, problem \eqref{vnf2-1} is equivalent to 
	\begin{equation}
		\label{vnf2-2}
		\begin{aligned} \revv{\minimize_{\substack{\x,\,\y,\, \z,\,\w}}}~&  \sum_{v \in \mathcal{V}}p_v y_v + \sum_{v\in\mathcal{V}} \sum_{k \in \mathcal{K}} \sum_{s \in \F^k} c^{k,s}_v x^{k,s}_v  \\
			{\revv{\text{subject to}}~}&  \eqref{onlyonenode}\text{--}\eqref{nodecapcons},~\eqref{ybounds},~\eqref{xbounds},~ \eqref{SFCineq-1},~ \eqref{SFCineq-3},~\eqref{validineq11}\text{--}\eqref{validineq6}.
		\end{aligned}
	\end{equation}
	We can further linearize constraints \eqref{validineq10} and \eqref{validineq20} as  
	\begin{align}
		& z^{k,s}_v \geq x^{k,s+1}_v - x^{k,s}_v,~\forall~v \in \mathcal{V}, ~\forall~s \in   \F^k \backslash\{\ell_k\},~\forall~k \in \mathcal{K}, \label{validineq3} \\
		& w^{k,s}_v \geq x^{k,s}_v - x^{k,s+1}_v, ~\forall~v \in \mathcal{V}, ~\forall~s \in   \F^k \backslash\{\ell_k\}, ~\forall~k \in \mathcal{K},\label{validineq4}
	\end{align}
	and obtain an equivalent \MBLP formulation as follows:
	\begin{equation}
		\label{vnf2}
		\tag{FP-II}
		\begin{aligned} \revv{\minimize_{\substack{\x,\,\y,\, \z,\,\w}}}~&  \sum_{v \in \mathcal{V}}p_v y_v + \sum_{v\in\mathcal{V}} \sum_{k \in \mathcal{K}} \sum_{s \in \F^k} c^{k,s}_v x^{k,s}_v  \\
			{\revv{\text{subject to}}~}&  \eqref{onlyonenode}\text{--}\eqref{nodecapcons},~\eqref{ybounds},~\eqref{xbounds},~ \eqref{SFCineq-1},~ \eqref{SFCineq-3},\\
			& \eqref{validineq11},~\eqref{validineq21},~\eqref{validineq5},~\eqref{validineq6},~\eqref{validineq3},~\eqref{validineq4}.
		\end{aligned}
	\end{equation}
	Indeed, constraints  \eqref{validineq10} and \eqref{validineq20} imply 
	\eqref{validineq3} and \eqref{validineq4}, respectively.
	On the other hand, due to the reason that variables $\{z_v^{k,s}\}$ and $\{w_v^{k,s}\}$ do not appear in the objective function of \eqref{vnf2}, there always exists an optimal solution of problem \eqref{vnf2} such that \eqref{validineq10} and \eqref{validineq20} hold.
	Therefore, problems \eqref{vnf2-2} and \eqref{vnf2} are equivalent.

\subsubsection{Effectiveness of derived inequalities}
To illustrate the effectiveness of inequalities \rev{\eqref{validineq1} and \eqref{validineq2} (or their linearization \eqref{validineq11}, \eqref{validineq21}, \eqref{validineq5}, \eqref{validineq6}, \eqref{validineq3}, and \eqref{validineq4})} in strengthening the \FP problem, let us consider a special case in which $\ell_k=1$ for all $k \in \K$ in \cite{Addis2020}. 
In this case,
the nonlinear terms $\sum_{s \in \F^k\backslash\{\ell_k\}}  \lambda_s^k {(x^{k,s+1}_v - x^{k,s}_v)}^+$ and  $\sum_{s \in \F^k \backslash\{\ell_k\}} \lambda_s^k {(x^{k,s}_v - x^{k,s+1}_v)}^+$ in \eqref{validineq1} and \eqref{validineq2}
can be removed, and thus we do not need to introduce variables $\{z_{v}^{k,s}\}$ and $\{w_{v}^{k,s}\}$.
Moreover, for this case, inequalities \eqref{nodecapcons}, \eqref{validineq1}, and \eqref{validineq2} can be combined into
\begin{equation}
	\label{s1validineq}
	\sum_{k \in \mathcal{K}} \lambda_0^k x^{k,1}_v \leq {\min}\left\{ \mu_v, \sum_{j  :  (j,v)\in \mathcal{L}} C_{jv}, \sum_{j : (v,j)\in \mathcal{L}} C_{vj} \right \}y_v.
\end{equation}
Notice that \eqref{s1validineq} is potentially much stronger than \eqref{nodecapcons}, especially when the overall capacity of node $v$'s incoming links or outgoing links is much less than the computational capacity of node $v$.
As such, problem \eqref{vnf2} can be potentially much stronger than problem \eqref{vnf1} in terms of providing a better relaxation of problem \eqref{ns}.
See the following example for an illustration.

\begin{example}
	Consider the toy example in Fig. \ref{example2}.
	Nodes $B$ and $C$ are cloud nodes whose computational capacities are $\mu_B=\mu_C =2$. 
	The link capacities are $C_{AB}=C_{CD} =1$ and $C_{AC}=C_{BD}=2$.
	There is only a single function $f$ and both the  cloud nodes can process this function. 
	Moreover, the power consumptions of placing the function at all cloud nodes are $0$ and 
	the power consumptions of activating the two cloud nodes are $p_B=1$ and $p_C = 2$, respectively.
	
	\begin{figure}[h]
		\begin{center}
			\begin{tikzpicture}[->,>=stealth',shorten >=0pt,auto,node distance=7cm,
				semithick]
				\tikzstyle{every state}=[fill=blue,draw=none,scale=0.25,text=white]	
				\node[state,label=above:{\footnotesize$A$}] (A)                  {};
				\node[state,label=above:{\footnotesize$B (2)$},rectangle]         (B) [above right = 0.7cm and 3cm of A] {};
				\node[state,label=above:{\footnotesize$C (2)$},rectangle]         (C) [below right  = 0.7cm and 3cm of A] {};
				\node[state,label=above:{\footnotesize$D$}]         (D) [right= 6cm of A] {};
				
				\path
				(A) edge             node {\footnotesize(1)} (B)
				(A) edge 			 node {\footnotesize(2)} (C)
				(B) edge 			 node {\footnotesize(2)} (D)
				(C) edge 			 node {\footnotesize(1)} (D)
				;
			\end{tikzpicture}
			\caption{A toy network example where the rectangular nodes are cloud nodes and the circular nodes are the source and destination of the services.
				$(a)$ and $(b)$ over each cloud node and each link denote that the node  and link capacities are $a$ and $b$, respectively.}
			\label{example2}
		\end{center}
	\end{figure}
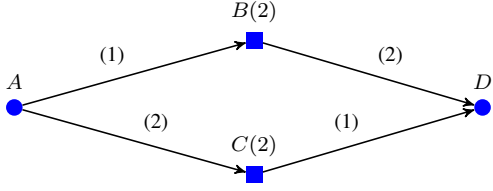
	Suppose that there are two services with their source and destination being $A$ and $D$, respectively, their SFC being $f$, and their data rates being $1$.
	Clearly, the objective value of problem \eqref{ns} is equal to $3$.
	However, solving problem \eqref{vnf1} will return a solution with the functions of both services being processed at cloud node $B$, yielding an objective value  $1$.
	Solving problem \eqref{vnf2} with the tighter constraints \eqref{s1validineq}, i.e., 
	\begin{align}
		\revvv{x_B^{1} + x_{B}^{2}} \leq {\min}\left\{ \mu_B, C_{AB}, C_{BD} \right \}y_B= y_B,\\
		\revvv{x_C^{1} + x_{C}^{2}} \leq {\min}\left\{ \mu_C, C_{AC}, C_{CD} \right \}y_C= y_C,
	\end{align}
	will yield a solution with an objective value  being $3$ (in which both cloud nodes are activated, as to process the  functions of the two services).
	This example clearly shows that adding constraints  \eqref{s1validineq} (or \rev{\eqref{validineq11}, \eqref{validineq21}, \eqref{validineq5}, \eqref{validineq6}, \eqref{validineq3}, and \eqref{validineq4}} for the general case) can effectively strengthen the \FP problem in terms of providing a better relaxation bound of problem \eqref{ns}.
\end{example}

The strongness of relaxation \eqref{vnf2} makes it much more suitable to be embedded into the proposed \rev{\cref{alg1}} than relaxation \eqref{vnf1}.
In particular, with this stronger relaxation \eqref{vnf2},  the proposed \rev{\cref{alg1}} converges much faster, thereby enjoying a much better overall performance; see \cref{subsec:ineq} further ahead.


%% file: section_numres2.tex
\section{Numerical Results}
\label{sect:numres}

In this section, we present numerical results to demonstrate the effectiveness and efficiency of the proposed valid inequalities and decomposition algorithm.
More specifically, we first perform numerical experiments to illustrate the effectiveness of the proposed valid inequalities for the proposed decomposition algorithm in \cref{subsec:ineq}.
Then, in \cref{subsec:proposedalg}, we present numerical results to demonstrate the effectiveness and efficiency of  the proposed \rev{\CBD} algorithm over  the \SOTA approaches in \rev{\cite{Chen2021,Chowdhury2012,Chen2023,Zhang2017}}.
%

{We use {CPLEX 20.1.0} to solve all LP and \MBLP problems. }
When solving the \MBLP problems, the time limit was set to 3600 seconds.
All experiments were performed on a server with 2 Intel Xeon E5-2630 processors and 98 GB of RAM, using the Ubuntu GNU/Linux Server {20.04} x86-64 operating system.


All algorithms were tested on the fish network \cite{Zhang2017}, which consists of 112 nodes and 440 links, including 6 cloud nodes. 
In order to test the effectiveness of connectivity based inequalities, 
we randomly remove some links so that there does not exist a path between some pair of two nodes in the network (i.e., $\V(v_0) \neq \emptyset$ for some $v_0 \in \V$). Each link is removed with probability $0.1$.
The cloud nodes' and links' capacities are randomly generated within $ [200,600] $ and $ [20,220] $, respectively.
For each service $k$, node $S(k)$ is randomly chosen from the available nodes and node $D(k)$ is set to be the common destination node; SFC $ \mathcal{F}(k) $ is a sequence of $4$ functions randomly generated from $ \{f^1,f^2, \ldots, f^5\} $; and $ \lambda^k_s $'s are the service function rates all of which are set to be the same integer value, chosen from $\{1, 2, \ldots, 40\}$.
The cloud nodes are randomly chosen to process $3$ functions of $ \{f^1, f^2, \ldots, f^5\} $.
The power consumptions of the activation of cloud nodes and the allocation of functions to cloud nodes are randomly generated with $[1, 200]$ and $[1,20]$, respectively.
The above parameters are carefully chosen to ensure that the constraints in problem \eqref{ns} are neither too tight nor too loose.
For each fixed number of services, $1000$ problem instances are randomly constructed  and the results reported below are averaged over these instances.

\subsection{Effectiveness of Proposed Valid Inequalities}
\label{subsec:ineq}

In this subsection, we illustrate the effectiveness of our proposed connectivity based inequalities  \eqref{SFCineq-1} and \eqref{SFCineq-3} and link-capacity based inequalities \rev{\eqref{validineq11}, \eqref{validineq21}, \eqref{validineq5}, \eqref{validineq6}, \eqref{validineq3}, and \eqref{validineq4}}.
To do this, we first compare the optimal values of
relaxations \eqref{vnf}, \eqref{vnf1}, and \eqref{vnf2}.
We compare the relative optimality gap improvement, defined by
\begin{equation}
	\label{eq1}
 \frac{\nu(\text{FP-x})-\nu(\text{FP})}{\nu(\text{NS})-\nu(\text{FP})}.
\end{equation}
where $\text{x} \in \{\text{I}, \text{II}\}$ and $\nu(\cdot)$ denotes the optimal objective value of the corresponding problem.
The gap improvement in \eqref{eq1} quantifies the improved tightness of relaxations \eqref{vnf1} and  \eqref{vnf2} over that of relaxation \eqref{vnf}, and thus 
the effectiveness of proposed connectivity based inequalities \eqref{SFCineq-1} and \eqref{SFCineq-3}, and link-capacity based inequalities  \rev{\eqref{validineq11}, \eqref{validineq21}, \eqref{validineq5}, \eqref{validineq6}, \eqref{validineq3}, and \eqref{validineq4}} in reducing the gap between the \FP problem and the \NS problem.
The larger the gap improvement, the stronger the relaxations \eqref{vnf1} and \eqref{vnf2} (as compared with relaxation \eqref{vnf}), and the more effective the proposed inequalities.
Note that the  gap improvement is a widely used performance measure in the integer programming community \cite{Vielma2010,Fukasawa2011} to show the tightness of a  relaxation problem over another one.

\begin{figure}[t]
	\centering
	\includegraphics[height=\figuresize]{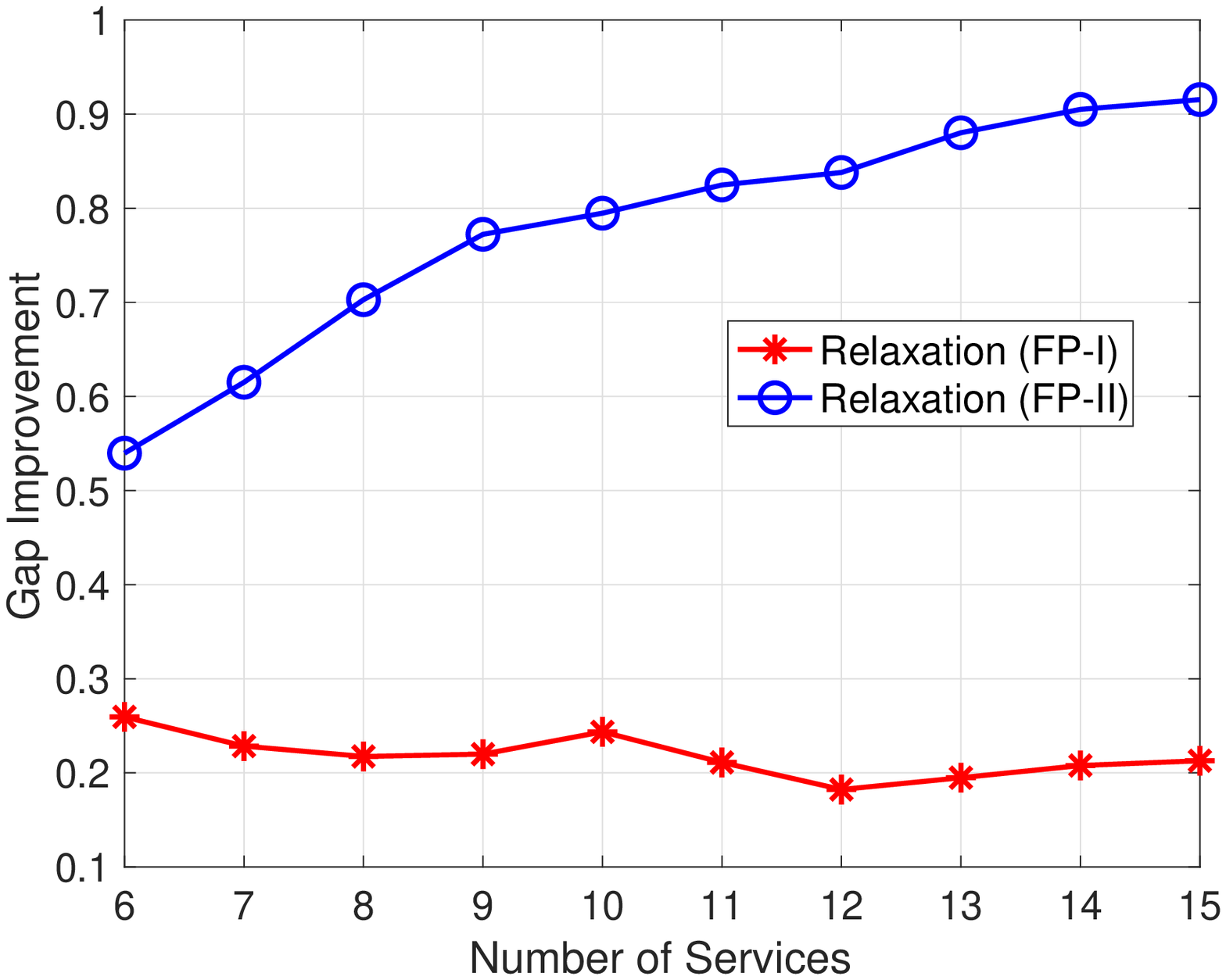}
	\caption{The average relative optimality gap improvement of relaxation  \eqref{vnf1} and \eqref{vnf2}.}
	\label{gap}
\end{figure}

\begin{figure}[t]
	\centering
	\includegraphics[height=\figuresize]{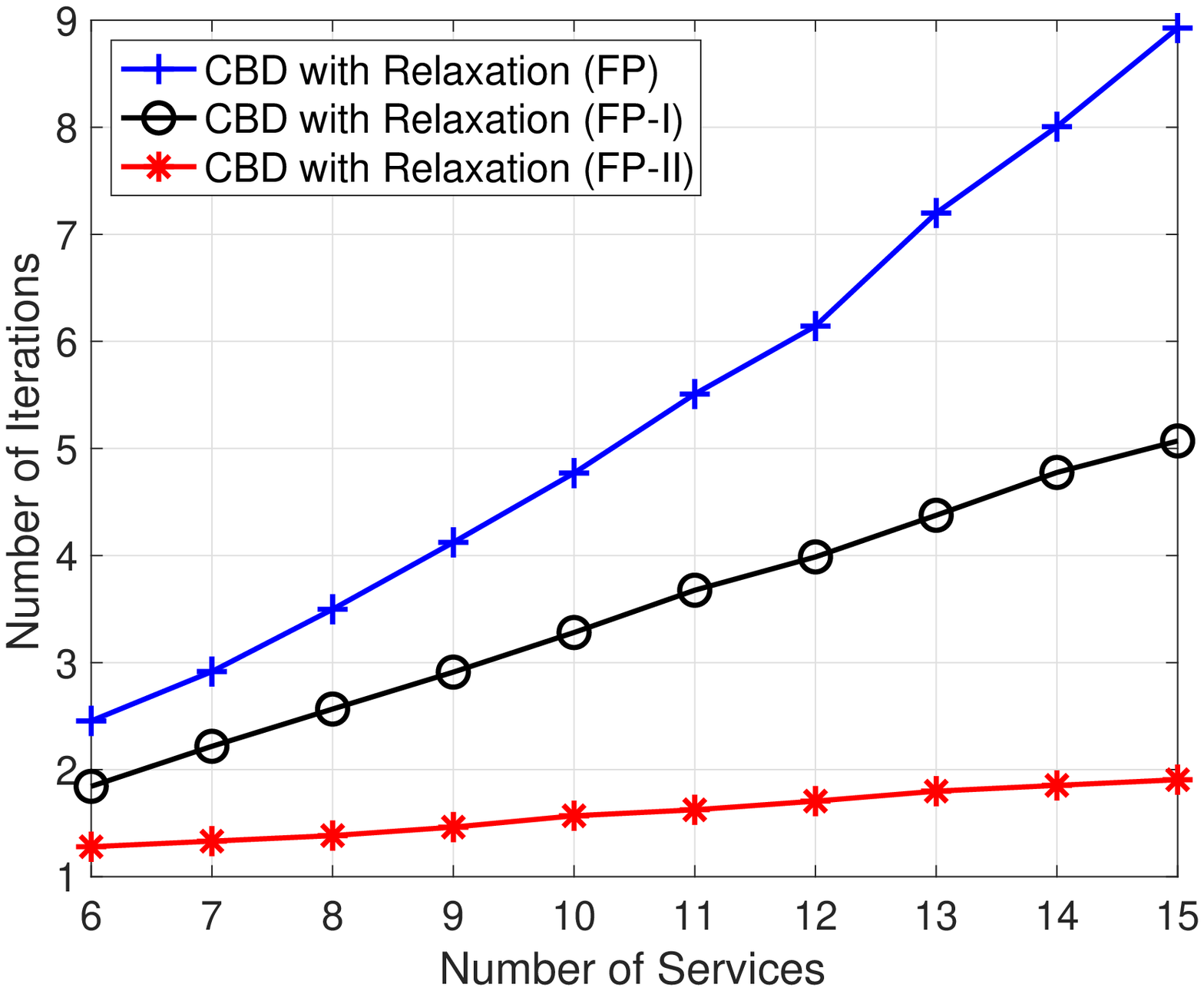}
	\caption{The average number of iterations needed by  \rev{\CBDs} with relaxations \eqref{vnf}, \eqref{vnf1}, and \eqref{vnf2}, respectively.}
	\label{iter}
\end{figure}

\begin{figure}[t]
	\centering
	\includegraphics[height=\figuresize]{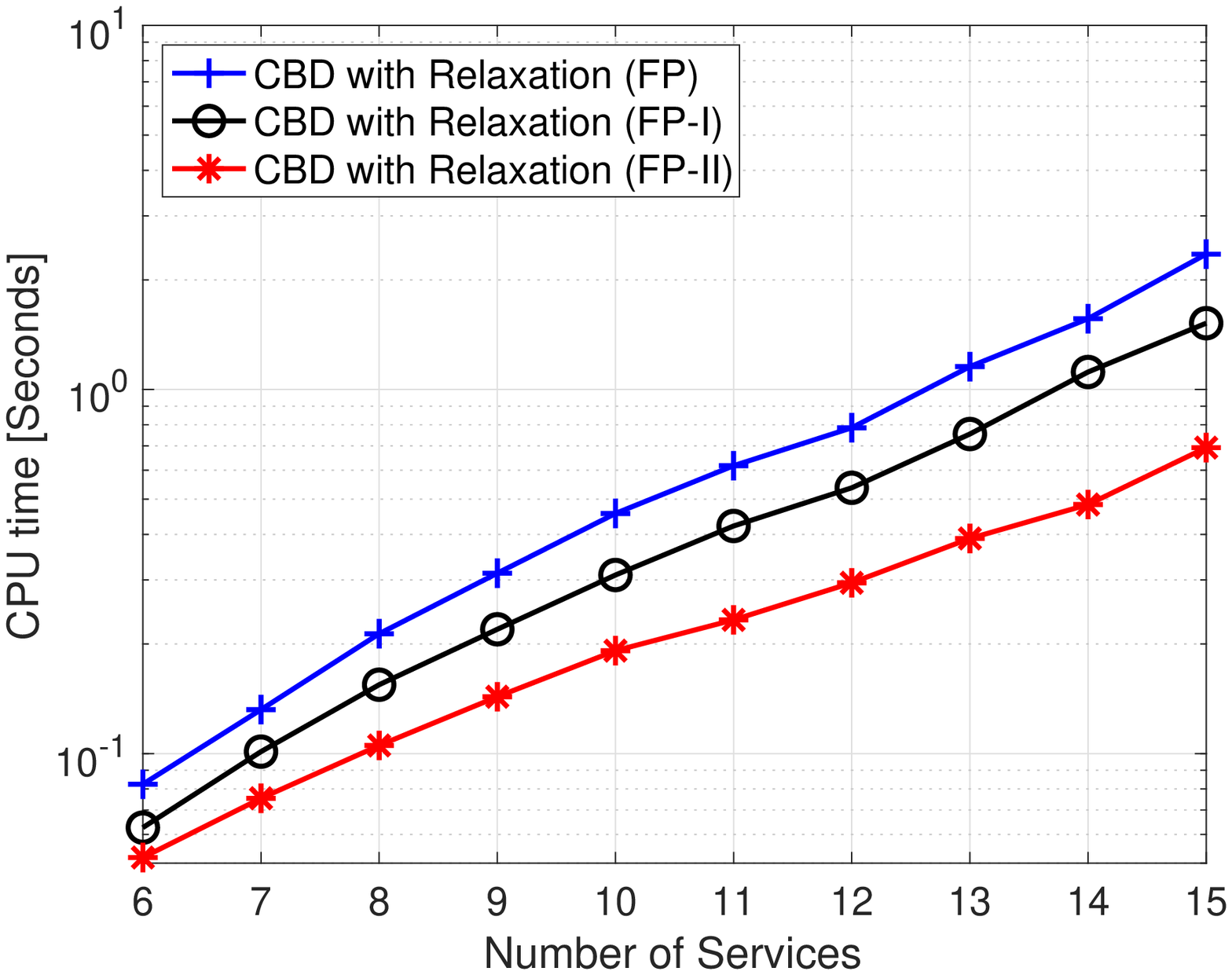}
	\caption{The average CPU time taken by   \rev{\CBDs}  with relaxations \eqref{vnf}, \eqref{vnf1}, and \eqref{vnf2}, respectively.}
	\label{BDtime}
\end{figure}

\cref{gap} plots the average  gap improvement versus different numbers of services.
As can be observed in \cref{gap}, in all cases, the gap improvement is larger than $0.1$, which shows that \eqref{vnf1} and \eqref{vnf2} are indeed stronger than \eqref{vnf}.
Compared with that of \eqref{vnf1}, the gap improvement of \eqref{vnf2} is much larger.
This indicates that the link-capacity based inequalities \rev{\eqref{validineq11}, \eqref{validineq21}, \eqref{validineq5}, \eqref{validineq6}, \eqref{validineq3}, and \eqref{validineq4}} can indeed significantly strengthening the FP problem only with the connectivity based inequalities \eqref{SFCineq-1} and \eqref{SFCineq-3}.
In addition, with the increasing of the number of services, the gap improvement of \eqref{vnf2} becomes larger. 
This can be explained as follows. 
As the number of services increases, the left-hand sides of inequalities \eqref{validineq5} and \eqref{validineq6} become larger, thereby making inequalities  \eqref{validineq5} and \eqref{validineq6} tighter. 
Consequently,   the gap improvement achieved by the link-capacity based inequalities is also likely to be larger.

Next, we illustrate the efficiency of the proposed connectivity based inequalities  \eqref{SFCineq-1} and \eqref{SFCineq-3} and link-capacity based inequalities \rev{\eqref{validineq11}, \eqref{validineq21}, \eqref{validineq5}, \eqref{validineq6}, \eqref{validineq3}, and \eqref{validineq4}} for accelerating the convergence of the \rev{\rev{\CBD}} algorithm (called \rev{\CBD}).
\cref{iter,BDtime} plot the average numbers of iterations and CPU times needed by the proposed \rev{\CBDs} with relaxations  \eqref{vnf}, \eqref{vnf1}, and \eqref{vnf2}, respectively, to terminate.
As observed, both of the derived families of valid inequalities can improve the performance of the proposed \rev{\CBD}. 
In particular, with the two proposed inequalities, the average number of iterations needed for the convergence of \rev{\CBD} is much smaller (less than $2$); and thus the CPU time taken by \rev{\CBD} is also much smaller.

From the above computational results, we can conclude that  the proposed \rev{\CBD} with relaxation \eqref{vnf2} (i.e., the \FP problem with the two families of valid inequalities) significantly outperforms that with relaxation \eqref{vnf} (i.e., the \FP problem without the two families of valid inequalities).
Due to this, we shall only test and compare the proposed \rev{\CBD} with relaxation \eqref{vnf2} with the other existing algorithms in the following.

\subsection{Comparison of \rev{\CBD} with \SOTA Algorithms}
\label{subsec:proposedalg}

\begin{figure}[t]
	\centering
	\includegraphics[height=\figuresize]{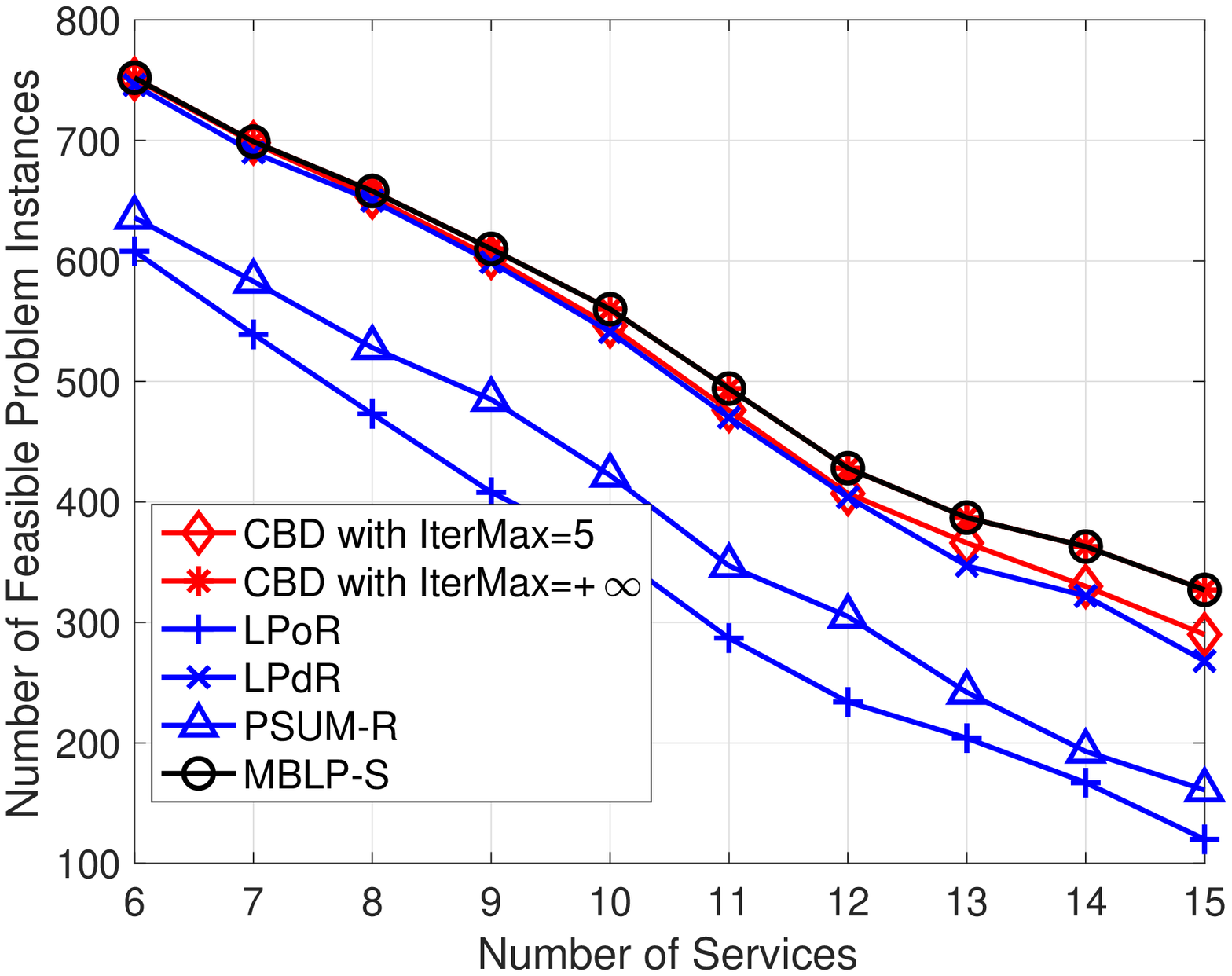}
	\caption{The number of feasible problem instances solved by \rev{\CBDs} with $\IterMax=5,\,+\infty$, \rev{\LPoR, \LPdR}, \PSUMR, and \EXACT.}
	\label{nfeas}
\end{figure}

\begin{figure}[t]
	\centering
	\includegraphics[height=\figuresize]{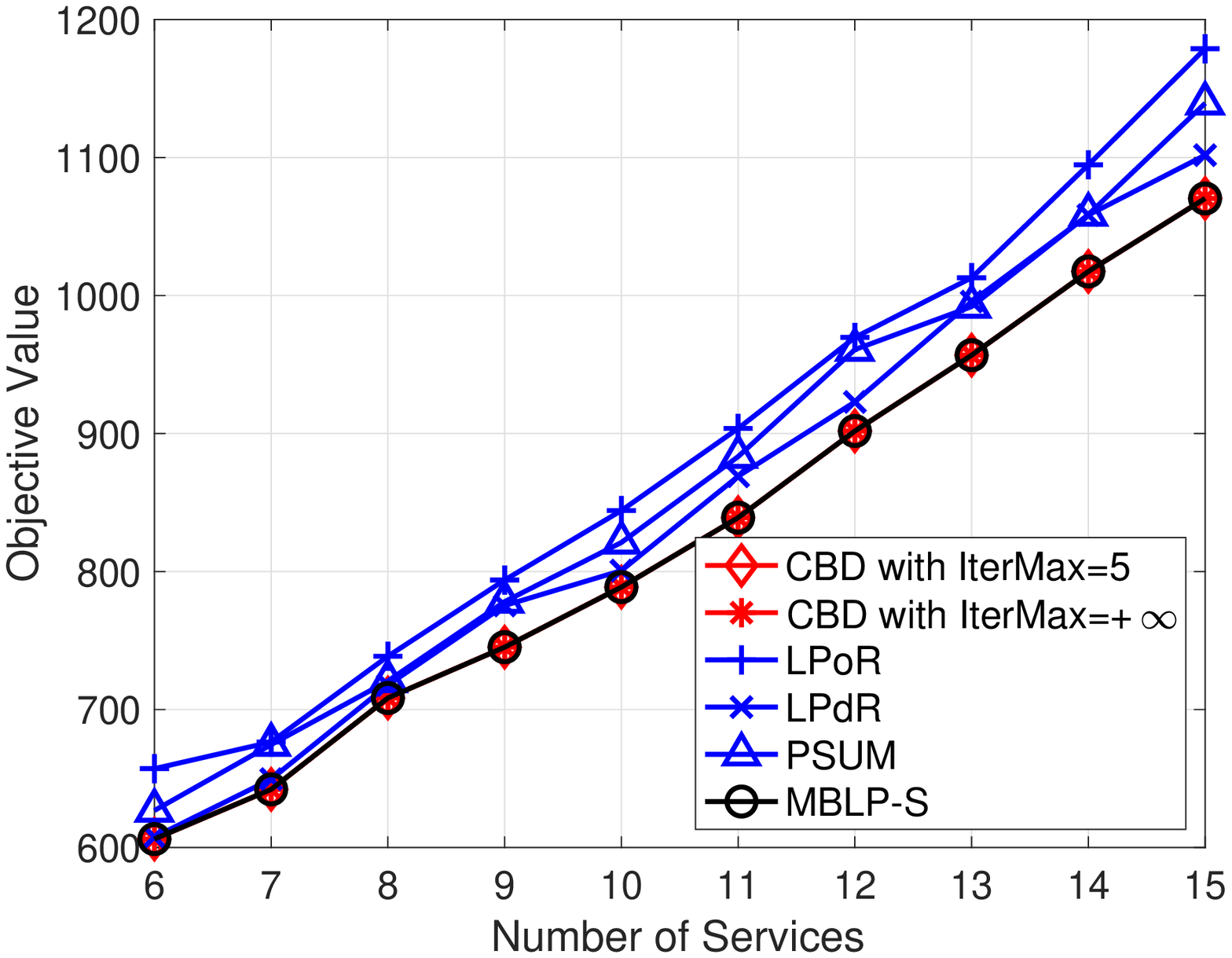}
	\caption{The average objective value returned by \rev{\CBDs} with $\IterMax=5,\,+\infty$, \rev{\LPoR, \LPdR}, \PSUMR, and \EXACT.}
	\label{obj}
\end{figure}

\begin{figure}[t]
	\centering
	\includegraphics[height=\figuresize]{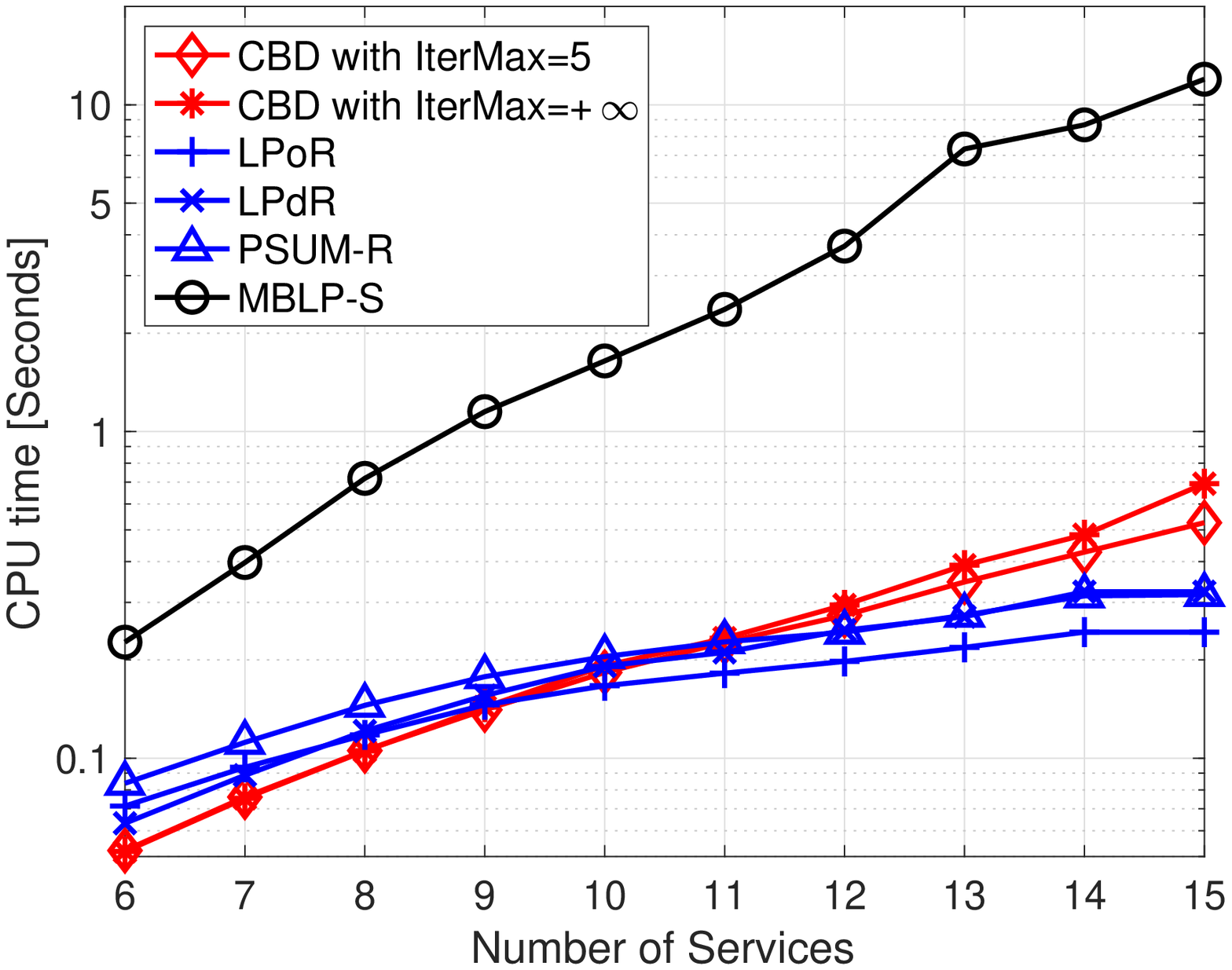}
	\caption{The average CPU time returned by \rev{\CBDs} with $\IterMax=5,\,+\infty$, \rev{\LPoR, \LPdR}, \PSUMR, and \EXACT.}
	\label{time}
\end{figure}

In this subsection, we compare the performance of the proposed \rev{\CBD algorithm} with \SOTA algorithms such as the exact approach using standard \MBLP solvers (called \EXACT) in \cite{Chen2021}, the LP \rev{one-shot} rounding (\rev{\LPoR}) algorithm in \cite{Chowdhury2012}, \rev{the LP dynamic rounding (\LPdR) algorithm in \cite{Chen2023}}, and the penalty successive upper bound minimization with rounding (\PSUMR)  algorithm in \cite{Zhang2017}.

\cref{nfeas,obj} plot the number of feasible problem instances and the average objective value returned by
the proposed \rev{\CBDs} with $\IterMax=5,\,+\infty$, \rev{\LPoR, \LPdR}, \PSUMR, and \EXACT, respectively.
When setting $\IterMax=5$, the proposed \rev{\CBD} cannot be guaranteed to find an optimal solution of problem \eqref{ns}, and the red-diamond curve in \cref{nfeas} is obtained as follows:
for a problem instance, if the proposed \rev{\CBD} can find an optimal solution of problem \eqref{ns} within 5 iterations, it is counted as a feasible instance; otherwise, it is infeasible.   
From \cref{nfeas,obj}, we can observe the effectiveness of the proposed \rev{\CBDs} over \rev{\LPoR, \LPdR}, and \PSUMR.
In particular, compared with  \rev{\LPoR, \LPdR}, and \PSUMR, the proposed  \rev{\CBDs} can find a feasible solution for much more problem instances \rev{(especially for instances with a large number of services)} and return a much better solution with a much smaller objective value.
{The latter is reasonable, as stated in \rev{\cref{subsec:analysis}}, \rev{\CBDs} must return an optimal solution as long as they find a feasible solution.}
Another simulation finding is that with the increasing of $\IterMax$, \rev{\CBD} can return a feasible solution for more problem instances.
Nevertheless, when $\IterMax=5$, \rev{\CBD} can return a feasible solution for \rev{most} instances (as \rev{\CBD} with $\IterMax=+\infty$ is able to find all truly feasible problem instances and the difference between the numbers of feasible problem instances solved by \rev{\CBDs} with $\IterMax=5$ and $\IterMax = +\infty$ is small in \cref{nfeas}).
This shows that in most cases, \rev{\CBD} can return an optimal solution within a few number of iterations, thanks to the two derived families of valid inequalities.

The comparison of the solution efficiency of the proposed \rev{\CBDs} with $\IterMax=5,\,+\infty$, \rev{\LPoR, \LPdR}, \PSUMR, and \EXACT is plotted in \cref{time}. 
We can observe from \cref{time} that though both \EXACT and the proposed \rev{\CBD} with $\IterMax=+\infty$ can return an optimal solution (if the problem is feasible),  \rev{\CBD} with $\IterMax=+\infty$ is much more computationally efficient than \EXACT.
In particular, the average CPU times taken by \rev{\CBD} with $\IterMax=+\infty$ are less than $0.7$ seconds, while the CPU times taken by \EXACT in the large-scale cases (i.e., $|\mathcal{K}|\geq 13$) are generally larger than $5$ seconds.
The solution efficiency of the proposed  \rev{\CBD} with $\IterMax= 5$,  \rev{\LPoR, \LPdR}, and  \PSUMR is comparable. 
In general, \rev{\LPoR} performs the best in terms of the CPU time, followed by \rev{\CBD} with $\IterMax=5$ \rev{, \LPdR, and \PSUMR}.

In summary, our simulation results show both of the effectiveness and efficiency of the proposed \rev{\CBD}. 
More specifically, compared with the \rev{\LPoR, \LPdR}, and \PSUMR algorithms in \cite{Chowdhury2012}, \rev{\cite{Chen2023} and} \cite{Zhang2017}, the proposed \rev{\CBD} is able to find a much better solution (due to the information feedback between the two subproblems);
compared with the exact approach in \cite{Chen2021}, the proposed \rev{\CBD} is significantly more computationally efficient, which is because of its decomposition nature.


\section{Conclusions}
\label{sect:conclusion}

In this paper, we have proposed an efficient \rev{\CBD} algorithm for solving large-scale \NS problems.
The proposed algorithm decomposes the large-scale \NS problem into two relatively easy \FP and \TR subproblems and solves them in an iterative fashion enabling a useful information feedback between them (to deal with the coupled constraints) until an optimal solution of the problem is found. 
To further improve the efficiency of the proposed  algorithm, we have also developed two families of valid inequalities to speed up its convergence by judiciously taking the special structure of the considered problem into account.
Two key features of the proposed algorithm, which make it particularly suitable to solve the large-scale \NS problems, are as follows. 
First, it is guaranteed to find \rev{a globally optimal} solution of the NS problem if the number of iterations between the two subproblems is allowed to be sufficiently large and it can return a much better solution than existing two-stage heuristic algorithms even though the number of iterations between the subproblems is small (e.g., 5).
Second, the \FP and \TR subproblems in the proposed  algorithm are a small \MBLP problem and a polynomial time solvable \LP problem, respectively, which are much easier to solve than the original problem. 
Simulation results show the effectiveness and efficiency of the proposed algorithm over the existing SODA algorithms in \red{\cite{Chen2021,Chowdhury2012,Chen2023,Zhang2017}}.
In particular, when compared with algorithms based on LP relaxations \red{\cite{Chowdhury2012,Chen2023,Zhang2017}}, the proposed algorithm is able to find a much better solution; when compared with the  algorithm needed to call \MBLP solvers \cite{Chen2021}, the proposed algorithm is much more computationally efficient.

\revv{
As the future work, it would be interesting to extend the proposed \CBD algorithm to solve the virtual network embedding (\VNE) problem in \cite{Vassilaras2017,Fischer2013}, which has wide applications in wireless and optical communication networks. 
The \VNE problem in \cite{Vassilaras2017,Fischer2013} is a variant of the \NS problem that requires the traffic flow between the two nodes hosting two adjacent functions to be routed via a \emph{single} path.
To develop a \CBD algorithm to solve the \VNE problem, we need to judiciously design an algorithm for solving the traffic routing subproblem (which is an integer programming problem), enabling to find a valid inequality to cut off the current infeasible solution.
}



%

%

%% file: appendix.tex
\appendices

\section{Proof of \cref{thm1}}
\label{appendixProof}

Since inequalities \eqref{SFCineq-1}--\eqref{SFCineq-2} include inequalities \eqref{SFCineq1}--\eqref{SFCineq3}, it follows that
$\X^2 \subseteq \X^1$ and $\X_{\LIN}^2 \subseteq \X_{\LIN}^1$.
Next, we prove $\X^3 \subseteq \X^2$ and $\X_{\LIN}^3 \subseteq \X_{\LIN}^2$.
First, adding \eqref{SFCineq-3} for all $s'=s, \ldots, s_0-1$ ($s_0 \geq s+1$) yields 
\begin{equation*}
	\sum_{s'=s}^{s_0-1}\sum_{v \in \V \backslash \V(v_0)} x^{k,s'}_v \leq 	\sum_{s'=s}^{s_0-1}\sum_{v \in \V \backslash \V(v_0)} x_v^{k,s'+1},
\end{equation*}
which is equivalent to 
\begin{equation}
	\sum_{v \in \V \backslash \V(v_0)} x^{k,s}_v \leq \sum_{v \in \V \backslash \V(v_0)} x_v^{k,s_0}. \label{tmpineq1}
\end{equation}
By \eqref{onlyonenode}, we have 
\begin{equation}\label{tmpineq2}
	\sum_{v \in \V} x_v^{k,s_0} = 1.
\end{equation}
Combining \eqref{tmpineq1} and \eqref{tmpineq2} shows 
\begin{equation}\label{tmpineq3}
	\sum_{v \in \V \backslash \V(v_0)} x^{k,s}_v + \sum_{v \in \V(v_0)} x_v^{k,s_0}\leq 1. 
\end{equation}
By the definition of $\V(v_0)$ in \eqref{v0def}, we have $v_0 \in \V \backslash \V(v_0)$. 
As a result, \eqref{tmpineq3} implies \eqref{SFCineq-2}, which shows  $\X^3 \subseteq \X^2$ and $\X_{\LIN}^3 \subseteq \X_{\LIN}^2$.

To show the theorem, it suffices to prove $\X^1 \subseteq \X^3$.
Next, we show the desirable result via verifying that \eqref{SFCineq-1} and \eqref{SFCineq-3} hold for any $\x \in \X^1$. 
We first show that \eqref{SFCineq-1} for $k \in \K$, $v \in \V(S^k)$, and $s \in \F^k$ hold at $\x$ by an induction of $s$.
By \eqref{SFCineq1}, $x_{v}^{k,s}=0$ holds for $s=1$.
Suppose that  $x_{v}^{k,s}=0$ holds for $s =1,2, \ldots,s_0$ ($1\leq s_0 < \ell_k$).
Then we have $ x_v^{k,s_0} = 0$ for all $v \in \V(S^k)$.
This, together with \eqref{onlyonenode}, implies
\begin{equation}
	\label{tmpineq4}
	\sum_{v \in \V \backslash \V(S^k)} x_v^{k,s_0} = 1.
\end{equation}
By \eqref{SFCineq3}, we have 
\begin{equation}
	x_{v_0}^{k,s_0} +  x_{v}^{k,s_0+1} \leq 1,\rev{~\forall~v \in \V(v_0),~\forall~v_0 \in \V\backslash \V(S^k)}.
\end{equation}
Combining the above with  \eqref{tmpineq4} and the fact that $x_{v_0}^{k,s_0} \in \{0,1\}$ for all $v_0 \in \V \backslash \V(S^k)$ gives 
\begin{equation}\label{tmpineq}
	x_v^{k,s_0+1} = 0,~\forall~v \in  \bigcap_{v_0 \in \V\backslash \V(S^k)}\V(v_0). 
\end{equation} 	
By \cref{fact} (ii), $\V(S^k) \subseteq \V(v_0)$ must hold for all $v_0 \in \V\backslash \V(S^k)$.
Hence, we have 
$$
\V(S^k) \subseteq \bigcap_{v_0 \in \V\backslash \V(S^k)}\V(v_0).
$$
Together with \eqref{tmpineq}, this implies that 	
$x_{v}^{k,s_0+1} = 0 $ holds for all $v \in \V(S^k)$.

Next, we show that \eqref{SFCineq-1} holds for all $k \in \K$, $v \in \V(D^k)$, and $s \in \F^k$ at $\x \in \X^1$ by an induction of $s$.
By \eqref{SFCineq2}, $x_{v}^{k,s}=0$ holds for  $s=\ell_k$.
Suppose that $x_{v}^{k,s}=0$ holds for $s =s_0, s_0 +1, s_0 +2,~\ldots,\ell_k$, where $1 < s_0 \leq \ell_k$.
Then we have $ x_v^{k,s_0} = 0$ for all $v \in \V(D^k)$.
This, together with \eqref{onlyonenode}, implies
\begin{equation}
	\label{tmpineq5}
	\sum_{v \in \V \backslash \V(D^k)} x_v^{k,s_0} = 1.
\end{equation}
By \eqref{SFCineq3}, we have 
\begin{equation}
	\label{tmpineq5-1}
	x_{v_0}^{k,s_0-1} + x_v^{k,s_0} \leq 1,\rev{~\forall~v \in \V(v_0),~\forall~v_0 \in \V(D^k)}.
\end{equation}
From \cref{fact} (iii), $\V\backslash \V(D^k) \subseteq \V(v_0)$ holds for all $v_0 \in \V(D^k)$.
This, together with \eqref{tmpineq5}, \eqref{tmpineq5-1}, and the fact that $ x_{v_0}^{k,s_0} \in \{0,1\}$ for all $v_0 \in \V \backslash \V(D^k)$, implies that  $x_{v_0}^{k,s_0-1}=0$ for all $v_0 \in \V(D^k)$.

Finally, we show that all inequalities in \eqref{SFCineq-3} hold at $\x \in \X^1$.
By \eqref{onlyonenode}, we have
\begin{align}
	& \sum_{v\in \V}x^{k,s+1}_v = 1,	\label{tmpineq7}
\end{align}
where $1 \leq s < \ell_k$.
For $v_0 \in \V$, by $x_{v_0}^{k,s} \in \{0,1\}$ and $x_v^{k,s+1} \in \{0,1\}$ for $v \in \V(v_0)$, 	\eqref{SFCineq3}, and \eqref{tmpineq7}, 
we have 
\begin{equation}
	\label{tmpineq8}
	x_{v_0}^{k,s} + \sum_{v \in \V(v_0)} x^{k,s+1}_v \leq 1.
\end{equation}
Combining \eqref{tmpineq7} and \eqref{tmpineq8} further yields 
\begin{equation}
	x_{v_0}^{k,s} \leq \sum_{v \in \V\backslash \V(v_0)} x^{k,s+1}_v.
\end{equation}
Similarly, for any $v' \in \V\backslash \V(v_0)$, we can show
\begin{equation}\label{tmpineq9}
	x_{v'}^{k,s} \leq \sum_{v \in \V\backslash \V(v')} x^{k,s+1}_v \leq \sum_{v \in \V\backslash \V(v_0)} x_v^{k,s+1},
\end{equation}
where the last inequality follows from the fact that $\V(v_0) \subseteq \V(v')$ for all  $v' \in \V\backslash \V(v_0)$, as stated in \cref{fact} (i).
Now we can use $\sum_{v' \in \V} x_{v'}^{k,s} = 1$ in \eqref{onlyonenode}, \eqref{tmpineq9}, and the binary natures of all variables $x_{v'}^{k,s} \in \{0,1\}$ to conclude the desirable inequalities in \eqref{SFCineq-3}.

\section{Details of \cref{example1}}
\label{appendixA}
In this part, we show that for the toy example in  \cref{example1}, solving the \LP relaxation of problem \eqref{vnf} with \eqref{SFCineq-1} and \eqref{SFCineq-3} will return a solution with the objective value being $\frac{1}{4}$, and 
solving the \LP relaxation of problem \eqref{vnf} with \eqref{SFCineq1}--\eqref{SFCineq3} or \eqref{SFCineq-1}--\eqref{SFCineq-2} will return a solution with the objective value being  $\frac{1}{6}$.

In this toy example, since there is only a single service, we let $x_{v}^s\in \{0,1\}$ denote whether function $f_s$ is allocated to cloud node $v$, where $v \in \{1,2,3\}$ and $s \in \{1,2\}$.
Since the power consumption of activating cloud nodes is zero, we can, without loss of generality, set $y_v=1$ in problem \eqref{vnf}, and hence all constraints in \eqref{xyxelation} are redundant. 
In this case, the constraints in \eqref{onlyonenode} reduce to 
\begin{align}
	& x_1^1 + x_2^1 + x_3^1 = 1,\label{cons1}\\
	& x_1^2 + x_2^2 + x_3^2 = 1.\label{cons2}
\end{align}
As $\mu_1=\mu_2 = +\infty$, $\mu_3=3$, and $y_1=y_2=y_3=1$, the only effective constraint in \eqref{nodecapcons} is
\begin{equation}
	2x_3^1 + 2x_3^2 \leq 3. \label{cons3}
\end{equation}
It is simple to check that $\V(S)=\V(D)=\emptyset$, $\V(1) = \emptyset$, $\V(2) = \{1\}$, and $\V(3)=\{1,2\}$, according to their definitions in \eqref{VSKdef}--\eqref{v0def}. 
Since $\V(S)=\V(D)=\emptyset$, constraints in \eqref{SFCineq1}, \eqref{SFCineq2}, and \eqref{SFCineq-1} do not exist.
As there are only two functions in the SFC of the service, constraints \eqref{SFCineq3} and \eqref{SFCineq-2} are the same (meaning that problem \eqref{vnf} with constraints \eqref{SFCineq1}--\eqref{SFCineq3} and that with \eqref{SFCineq-1}--\eqref{SFCineq-2} are also equivalent), which are 
\begin{align}
	& x_2^1 + x_1^2 \leq 1, \label{exineq1}\\
	& x_3^1 + x_1^2 \leq 1, \label{exineq2}\\
	& x_3^1 + x_2^2 \leq 1.\label{exineq3}
\end{align}  
Finally, all constraints in \eqref{SFCineq-3} in this toy example read
\begin{align}
	& x_1^1 + x_2^1 + x_3^1 \leq x_1^2 + x_2^2 + x_3^2, \label{exineq4}\\
	& x_2^1 + x_3^1 \leq x_2^2 + x_3^2, \label{exineq5}\\
	& x_3^1 \leq  x_3^2.\label{exineq6}
\end{align}
Therefore, the LP relaxation problem \eqref{vnf} with \eqref{SFCineq1}--\eqref{SFCineq3} or \eqref{SFCineq-1}--\eqref{SFCineq-2}  can be presented as 
	\begin{align}
		z_{\LP_1}=\revv{\minimize_{\x}} ~&  x_1^1 + x_2^1 \nonumber\\
		{\revv{\text{subject to}}~} & \eqref{cons1}\text{--}\eqref{exineq3},\label{exLP1}\\
		& x_v^s \in [0,1], ~\forall~v \in \{1,2,3\}, ~\forall~s \in \{1,2\},\nonumber
	\end{align}
and
the LP relaxation problem \eqref{vnf} with \eqref{SFCineq-1} and \eqref{SFCineq-3} can be presented as 
\begin{align}
	z_{\LP_2}=\revv{\minimize_{\x}} ~&  x_1^1 + x_2^1 \nonumber\\
		{\revv{\text{subject to}}~} & \eqref{cons1}\text{--}\eqref{cons3},~ \eqref{exineq4}\text{--}\eqref{exineq6},	\label{exLP2}\\
		& x_v^s \in [0,1], ~\forall~v \in \{1,2,3\}, ~\forall~s \in \{1,2\}. \nonumber
\end{align}
Below we show that $z_{\LP_1}=\frac{1}{6}$ and $z_{\LP_2}=\frac{1}{4}$, separately.\vspace{0.1cm}\\
{\noindent$\bullet$~Proof of $z_{\LP_1}=\frac{1}{6}$.\vspace{0.1cm}\\}
Let $\bar{\x}$ be given as 
\begin{align*}
	& \bar{x}_1^1 = \frac{1}{6},~\bar{x}_2^1 = 0,~\bar{x}_3^1 = \frac{5}{6}, ~ \bar{x}_1^2 = \frac{1}{6},~\bar{x}_2^2 =  \frac{1}{6},~\bar{x}_3^2 = \frac{2}{3}. 
\end{align*}
It is simple to check that $\bar{\x}$  is feasible to problem \eqref{exLP1}, and thus $z_{\LP_1}\leq \frac{1}{6}$.
It remains to show $x_1^1 + x_2^1 \geq \frac{1}{6}$ for any feasible solution $\x$ of \rev{problem} \eqref{exLP1}.
By carefully combining \eqref{cons3}, \eqref{exineq2}, and \eqref{exineq3}, we have 
\begin{equation*}
	\begin{aligned}
	7 & \geq (2x_3^1 + 2x_3^2)+  2(x_3^1 + x_1^2) + 2(x_3^1 + x_2^2) \\
	&  = 6x_3^1 + 2(x_1^2+x_2^2+x_3^2)= 6x_3^1 + 2,
	\end{aligned}
\end{equation*}
where the last equality follows from \eqref{cons2}.
Hence, we have $x_3^1 \leq \frac{5}{6}$, which, together with \eqref{cons1}, implies $x_1^1 + x_2^1 \geq \frac{1}{6}$. 
\vspace{0.1cm}\\
{\noindent$\bullet$~Proof of $z_{\LP_2}=\frac{1}{4}$.\vspace{0.1cm}\\}
Let $\hat{\x}$ be given as 
\begin{align*}
	& \hat{x}_1^1 = \frac{1}{4},~\hat{x}_2^1 = 0,~\hat{x}_3^1 = \frac{3}{4}, ~ \hat{x}_1^2 = \frac{1}{4},~\hat{x}_2^2 = 0,~\hat{x}_3^2 = \frac{3}{4}.
\end{align*}
It is simple to check that $\hat{\x}$  is feasible to problem \eqref{exLP2}, and thus $z_{\LP_2}\leq \frac{1}{4}$.
It remains to show $x_1^1 + x_2^1 \geq \frac{1}{4}$ for any feasible solution $\x$ of \rev{problem} \eqref{exLP2}.
By \eqref{cons3} and \eqref{exineq6}, we have $3 \geq 2x_{3}^1+ 2x_{3}^2 \geq  2x_{3}^1+  2x_{3}^1=4x_{3}^1$ and thus $x_3^1 \leq \frac{3}{4}$, 
which, together with \eqref{cons1}, implies $x_1^1 + x_2^1 \geq \frac{1}{4}$.

%% file: references.tex
\def\longreference{0}
\ifthenelse{\longreference = 1}{
\newcommand{\JCST}{{IEEE communications surveys \& tutorials}\xspace}
\newcommand{\JTNSM}{{IEEE Transactions on Network and Service Management}\xspace}
\newcommand{\JTN}{{IEEE/ACM Transactions on Networking}\xspace}
\newcommand{\JSAC}{IEEE Journal on Selected Areas in Communications\xspace}
\newcommand{\JTSP}{IEEE Transactions on Signal Processing\xspace}
\newcommand{\JCCR}{Computer Communication Review\xspace}
\newcommand{\JCE}{Chinese Journal of Electronics\xspace}
\newcommand{\JTWC}{IEEE Transactions on Wireless Communications\xspace}
\newcommand{\JTMC}{IEEE Transactions on Mobile Computing\xspace}
\newcommand{\JIJCS}{International Journal of Communication Systems\xspace}
\newcommand{\JOR}{Operations Research\xspace}
\newcommand{\JMP}{Mathematical Programming\xspace}
\newcommand{\JICM}{IEEE Communications Magazine\xspace}
\newcommand{\JICST}{IEEE Communications Surveys \& Tutorials}

\newcommand{\PROC}{Proceedings of\xspace}

\newcommand{\JAN}{{January}\xspace}
\newcommand{\FEB}{Feburary\xspace}
\newcommand{\MAR}{March\xspace}
\newcommand{\APR}{April\xspace}
\newcommand{\MAY}{May\xspace}
\newcommand{\JUN}{June\xspace}
\newcommand{\JUL}{July\xspace}
\newcommand{\AUG}{August\xspace}
\newcommand{\SEP}{September\xspace}
\newcommand{\OCT}{October\xspace}
\newcommand{\NOV}{November\xspace}
\newcommand{\DEC}{December\xspace}
}
{
\newcommand{\JCST}{{IEEE Commun. Surv. Tut.}\xspace}
\newcommand{\JTNSM}{{IEEE Trans. Netw. Service Manag.}\xspace}
\newcommand{\JTN}{{IEEE/ACM Trans. Netw.\xspace}}
\newcommand{\JSAC}{IEEE J. Sel. Areas Commun.\xspace}
\newcommand{\JTSP}{IEEE Trans. Signal Process.\xspace}
\newcommand{\JCCR}{Comput. Commun. Rev.\xspace}
\newcommand{\JCE}{Chin. J. Electron.\xspace}
\newcommand{\JTWC}{IEEE Trans. Wirel. Commun.\xspace}
\newcommand{\JTMC}{IEEE Trans. Mob. Comput.\xspace}
\newcommand{\JIJCS}{Int. J. Commun. Syst.\xspace}
\newcommand{\JOR}{Oper. Res.\xspace}
\newcommand{\JMP}{Math. Program.\xspace}
\newcommand{\JICM}{IEEE Commun. Mag.\xspace}
\newcommand{\JICST}{IEEE Commun. Surv. Tutor.}

\newcommand{\PROC}{Proc.\xspace}

\newcommand{\JAN}{Jan.\xspace}
\newcommand{\FEB}{Feb.\xspace}
\newcommand{\MAR}{Mar.\xspace}
\newcommand{\APR}{Apr.\xspace}
\newcommand{\MAY}{May\xspace}
\newcommand{\JUN}{Jun.\xspace}
\newcommand{\JUL}{Jul.\xspace}
\newcommand{\AUG}{Aug.\xspace}
\newcommand{\SEP}{Sep.\xspace}
\newcommand{\OCT}{Oct.\xspace}
\newcommand{\NOV}{Nov.\xspace}
\newcommand{\DEC}{Dec.\xspace}

}

\bibliographystyle{IEEEtran}